\documentclass[10pt, journal]{IEEEtran}
\IEEEoverridecommandlockouts
\usepackage{cite}
\usepackage[binary-units,per-mode=symbol]{siunitx}
\usepackage{amsmath,amssymb,amsfonts}
\usepackage[ruled,vlined, linesnumbered]{algorithm2e}
\usepackage{graphicx}
\usepackage{subfigure}
\usepackage{textcomp}
\usepackage{xcolor}
\usepackage{dsfont}
\usepackage{thmtools, thm-restate}
\usepackage{amsthm}

\begin{document}
%\bstctlcite{IEEEexample:BSTcontrol}
\title{DUNE: Improving Accuracy for Sketch-INT Network Measurement Systems}

\author{Zhongxiang~Wei,~\IEEEmembership{Student~Member,~IEEE,}
        Ye~Tian,~\IEEEmembership{Member,~IEEE,}
        Wei~Chen,~
        Liyuan~Gu,~
        and~Xinming~Zhang,~\IEEEmembership{Senior~Member,~IEEE}% <-this % stops a space
\IEEEcompsocitemizethanks{\IEEEcompsocthanksitem The authors are with the Anhui Key Laboratory of High Performance Computing, School of Computer Science and Technology, University of Science and Technology of China, Hefei, Anui, 230026, China.\protect\\
E-mail: \{wz199758, szcw33, guliyuan1\}@mail.ustc.edu.cn, \{yetian, xinming\}@ustc.edu.cn
}% <-this % stops an unwanted space
}
\maketitle
%\IEEEtitleabstractindextext{

\begin{abstract}

%Monitoring network health in real time is critical for today's production network, and among the solutions,
In-Band Network Telemetry (INT) and sketch algorithms are two representative methodologies for measuring network traffics in real time.
%achieves high accuracy but incurs a large network overhead; while the sketch-based solution is flexible by providing a tradeoff between accuracy and resource usage, but requires an out-band channel to send sketches to analyzers.
To combine sketch with INT and to keep their advantages, the ``reconstructing sketch at end-host'' approach, which uses INT to send small pieces of switch sketch (i.e., \emph{sketchlet}) to end-host for reconstructing an identical sketch, is a promising direction. However, we reveal that the naive column sketchlet is not efficient, and inaccuracies arise because of the invalid and stale measurement data in the end-host reconstructed sketch.
In this paper, we present an innovative sketch-INT measurement system named \emph{DUNE}.
DUNE is lightweight by following the ``reconstructing the sketch at end-host'' approach, and to improve the measurement accuracy, we make two innovations: First, we design a novel sketchlet named \emph{scatter sketchlet} that is more efficient in transferring measurement data by allowing a switch to select individual sketch buckets to add to sketchlet; Second, we develop data structures for tracing ``freshness'' of the sketch buckets, and present algorithms for smartly selecting buckets that contain valuable measurement data to send to end-host. We theoretically prove that our proposed methods are superior comparing with the existing solution, and implement a prototype on the commodity Barefoot Tofino switches. We conduct extensive experiments on DUNE, and the evaluation results show that the system considerably improves measurement accuracies at negligible costs. In particular,
with less than $0.36\%$ loss of the packet forwarding rate, DUNE avoids up to $60\%$ errors in the end-host reconstructed sketch.

\end{abstract}


\section{Introduction}

\IEEEPARstart{T}oday's production network is composed of a large number and variety of network elements including routers, switches, and middleboxes. In such a network, faults and errors may arise from any single or combination of these elements, therefore, how to monitor network health in real time is a critical problem in network management.

With the advances of software-defined networking (SDN) and data plane programmability, a number of measurement-based solutions have been proposed for troubleshooting networks in recent years. Among them, one promising direction is In-band Network Telemetry (INT) \cite{JeyakumarINT14}\cite{KimINTdemo15}\cite{TammanaPath16}\cite{INTSpec20}. In INT, a programmable switch piggybacks measurement data in packet header, and sends them to end-host for analysis. The benefit of INT is its accuracy, as per-packet information is collected. However, by carrying measurement data, INT consumes extra bandwidth, thus considerably impacts an INT flow's goodput and completion time \cite{KimSelectiveINT18}\cite{ShengDeltaINT21}\cite{SongINTlabel21}\cite{BasatPINT20}.

Another promising direction is to develop measurement systems based on sketches \cite{YuOpenSketch13}\cite{LiFlowRadar16}\cite{YangElastic18}\cite{YangDiamond19}\cite{LiuSFSketch20}\cite{ZhangCocoSketch21}\cite{HuangSketchVisor17}\cite{HuangCompressiveSense21}\cite{LiuUnivMon16}. In such a system, a probabilistic data structure, namely \emph{sketch}, is maintained by switch for aggregating per-packet information. A sketch-based system is flexible as it provides a tradeoff between accuracy and resource usage. However, to send sketches to analyzers, an out-of-band channel is required, which either demands a dedicated channel \cite{YuOpenSketch13}\cite{LiFlowRadar16}, or is greatly impacted by the available shared bandwidth \cite{YangElastic18}.

Since both INT and sketch-based methods have pros and cons, people start to consider combining them and keep their advantages. There are two representative approaches. The first approach, with SketchINT \cite{YangSketchINT21} as an example, is to ``construct sketch at end-host''. In such an approach, a network switch piggybacks packet-level information using INT, and when receiving an INT packet, the end-host aggregates the data into a group of sketches. A benefit of this approach is that unlike the switch with limited programmability, an end-host can maintain novel and complex sketch structures that are difficult to be implemented on switch hardwares.

The second approach, with LightGuardian \cite{ZhaoLightGuardian21} as an example, is to ``reconstruct sketch at end-host''. In this approach, a sketch is maintained by a programmable switch, and exploits the switch's visibility to trace per-flow traffic statistics. Moreover, the switch splits the sketch structure into many small pieces, called \emph{sketchlets}, and an INT flow carries the sketchlets to end-host, which resembles them to reconstruct a sketch that has an identical structure as the one on the switch.
By making use of INT, out-band channel is no long required, moreover, with the abundant computation and memory resources, an end-host can maintain many sketch instances for tracing traffic characteristics over a long time, and provide computational-intensive query services.

%which is infeasible on hardware switches with limited memories.

Unfortunately both approaches have their limitations. The ``constructing sketch at end-host'' approach achieves high accuracy, as end-host directly collects packet-level information, but it does not avoid the large network overhead of INT. The ``reconstructing sketch at end-host'' approach effectively reduces the INT bandwidth usage, as a sketchlet contains flow-level statistics that have already been aggregated by the switch sketch. However, as we will see in this paper, the reconstructed sketch is not as accurate as the switch sketch, for the reason that the two sketches are not timely synchronized, and when data in the reconstructed sketch is invalid or stale, inaccuracies are introduced.

In this paper, we present a sketch-INT network measurement system named \emph{DUNE}. We have three design objectives:

\begin{itemize}
  \item \emph{Lightweight}: The system should be lightweight regarding network overhead by following the ``reconstructing sketch at end-host'' approach in combining sketch with INT.
  \item \emph{Accurate}: The system should provide accurate measurement results by reducing the inaccuracies caused by invalid and stale data in the end-host reconstructed sketch.
  \item \emph{Practical}: The system should be practical regarding realization. In particular, the key components of the system should be designed under the constraints of the RMT (Reconfigurable Match-Action Tables) programmable switches (e.g., the Intel Tofino switch) \cite{BosshartRMT13}\cite{CholedRMT17}.
\end{itemize}

We make four contributions to fulfill the objectives in this paper:
\begin{itemize}
  \item We follow the ``reconstructing sketch at end-host'' approach by using INT to send sketchlets to end-host. Moreover, we present a novel sketchlet design named \emph{scatter sketchlet}, which allows a switch to select individual sketch buckets to add to sketchlet. We prove in theory that the scatter sketchlet is more efficient in transferring measurement data to end-host than the existing approach.
  \item We develop data structures for tracing ``freshness'' of the sketch buckets, and present algorithms for selecting buckets that contain valuable measurement data to sketchlets. We theoretically prove that our proposed methods achieve the desired property in selecting sketch buckets at the frequencies that are proportional to their update frequencies.
  \item We realize our proposed data structures and algorithms on the P4 Tofino switch under the device's strict register access constraints, and make our implementation open-source.
  \item We carry out extensive evaluations on both the software and hardware implementations, and find that DUNE significantly improves measurement accuracy at negligible cost. In particular, DUNE avoids $60\%$ measurement errors, but only slightly reduces a Tofino switch's forwarding rate less than $0.36\%$.
\end{itemize}

The remainder part of this paper is organized as the following. We discuss the related works in Sec. \ref{Sec_Related}. Sec. \ref{Sec_MotiveOverview} explains our motivation and presents an overview of the DUNE system. Sec. \ref{Sec_Scatter} presents the design and analysis of the scatter sketchlet. We present the sketch bucket selection algorithms in Sec. \ref{Sec_SelectAlg} and describe the prototype implementation in Sec. \ref{Sec_Implement}. Sec. \ref{Sec_Eval} discusses the evaluation results and we conclude in Sec. \ref{Sec_Conclude}.

\section{Related Work}\label{Sec_Related}

Conducting comprehensive measurements for monitoring large-scaled networks is challenging. In the following, we introduce the major categories of representative solutions. For other solutions, please refer to references \cite{HandigolNetSight14}\cite{NarayanaPath16}\cite{TammanaSwitchPointer18}\cite{Sonchack*Flow18}\cite{HuangOmniMon20}.

\noindent\textbf{Sampling-based solutions.}
%Today's production networks generally contain a wide range and variety of network elements including routers, switches, and middleboxes from different vendors. Faults and errors are frequent in production networks. To troubleshooting network faults in real time,
%Sampling-based methods are widely applied for troubleshooting networks.
In sampling-based measurement systems such as NetFlow \cite{NetFlow}, sFlow \cite{sFlow}, and Everflow \cite{ZhuPacketLevel15}, filtering rules are set up to collect traffics that satisfy certain conditions for analysis. Although sampling-based methods have been successfully applied for decades, however, it is doubtful whether accurate results can be derived from only a subset of the traffic.

\noindent\textbf{Probing-based solutions.}
Probing-based approaches are also widely used. For example, Pingmesh \cite{GuoPingmesh15} analyzes abnormals based on a probe-based latency measurement. NetBouncer \cite{TanNetBouncer19} detects device and link failures by actively probing paths in data center networks. The limitation of the probing-based method is that only the probe traffic is measured.

\noindent\textbf{In-band Network Telemetry (INT)-based solutions.}
With the advances of software-defined networking (SDN) and programmable data plane, In-band Network Telemetry (INT), which collects per-packet information with ordinary network flows, becomes a promising direction in recent years. Over the OpenFlow data plane, PathDump \cite{TammanaPath16} traces per-packet trajectories and provide a set of APIs for analyzers to debug networks. Jeyakumar et al. \cite{JeyakumarINT14} propose to allow programmable switches to execute ``tiny packet programs'' (TPPs) embedded in packets to collect per-packet network states.
Kim et al. \cite{KimINTdemo15} demonstrate that INT can be realized on the P4-programmable data plane, and \texttt{P4.org} develops the technical specification for supporting INT over the P4 data plane \cite{INTSpec20}. A major concern of INT is that by carrying measurement data in packet header, INT consumes considerable extra bandwidth, and its overhead increases with network size. To reduce the overhead, Kim et al. \cite{KimSelectiveINT18} propose to adjust the insertion ratio to carry only the significant changes of the monitored network states. Sheng et al. \cite{ShengDeltaINT21} present DeltaINT, which reduces INT overhead by selectively carrying network states only when their values change significantly. Song et al. \cite{SongINTlabel21} propose to insert device states to packet headers based on dynamically adjusted intervals. Basat et al. \cite{BasatPINT20} propose PINT, which applies various probabilistic techniques to encode measurement data on multiple packets, so as to reduce the per-packet INT overhead.

\noindent\textbf{Sketch-based solutions.}
Sketch-based method is another promising direction. In a sketch-based system, usually a probabilistic data structure, namely sketch, is maintained within a switch for aggregating per-packet information. Representative sketches include bitmap \cite{EstanBitmap06}, hashing table \cite{SongHashtable05}, count-min \cite{CormodeCM05}, Bloom filter \cite{BonomiBF06}, and their variants. Many works focus on generalizing and optimizing sketch algorithms.
%Yu et al. \cite{YuOpenSketch13} define a software-defined traffic measurement architecture named OpenSketch, and design a three-stage pipeline that contains various sketches to support different measurement tasks.
Yang et al. \cite{YangElastic18} present a generic sketch named Elastic Sketch that identifies and differentiates large flows from small ones, and is adaptive to traffic variances. To adapt to skewed network flows, Yang et al. \cite{YangDiamond19} propose a novel sketch, namely the Diamond sketch, to dynamically assign appropriate amount of resources to each flow on demand. Liu et al. \cite{LiuSFSketch20} propose a new sketch called the Slim-Fat (SF) sketch that improves high accuracy without sacrificing the update and query speed.
%Zhou et al. \cite{ZhouGeneralSketch19} propose a set of common frameworks, each for a family of traffic measurement tasks that share the same implementation structure.
Song et al. \cite{SongFCM20} propose a tree-based sketch structure named FCM-sketch as a more accurate and memory-efficient replacement of the count-min sketch.
Zhang et al. \cite{ZhangCocoSketch21} design a structure named CocoSketch that is capable to support partial key queries. Huang et al. \cite{HuangSketchVisor17}\cite{HuangCompressiveSense21} apply compressive sensing to recover measurement results from errors. For optimizing sketch-based measurement systems from a network-wide view, OpenSketch \cite{YuOpenSketch13} derives the sketch parameters on individual switches by solving an optimization problem; UnivMon \cite{LiuUnivMon16} dispatches measurement tasks to sketches hosted on different switches by solving an integer programming problem.
%In the OmniMon system [], the sampling, storing, computing, and transmitting workloads in network measurements are schedule to different types of network elements based on their resources and visions.

\noindent\textbf{Combining Sketch with INT.}
%Both INT and sketch-based methods have pros and cons.
%The benefit of INT is its accuracy, as per-packet information is collected. But INT consumes a large volume of extra bandwidth, and its network overhead increases with network size. A sketch-based system is usually lightweight as a trade off can be made between estimating accuracy and sketch size. However, to send sketches to end-host analyzers, an out-of-band channel is required, which either demands a dedicated channel \cite{YuOpenSketch13}\cite{LiFlowRadar16}, or is greatly impacted by the available shared bandwidth \cite{YangElastic18}.
To combine sketch with INT, Yang et al. \cite{YangSketchINT21} design a novel sketch named TowerSketch, and use it at the network edge to aggregate per-packet INT information. Zhao et al. \cite{YangSketchINT21} design a novel sketch named SuMax on switch, and divide it into small-sized sketchlets to send to end-host using INT; on receiving the sketchlets, the end-host reconstructs a sketch that has an identical structure as the one on the switch, to provide query services.

\section{Motivation and System Overview} \label{Sec_MotiveOverview}

\subsection{Motivation}\label{Sec_Motivation}

In this work, we focus on the lightweight ``reconstructing sketch at end-host'' approach, and aim to improve its measurement accuracy. Our work is inspired by LightGuardian \cite{ZhaoLightGuardian21}, which we briefly introduce as the following.

In LightGuardian, a programmable switch maintains a sketch structure named SuMax, which can be viewed as a modified count-min sketch composed of $w$ columns and $d$ rows of buckets.
%Each bucket contains multiple counters for tracing different traffic characteristics.
A sketchlet in LightGuardian is simply a column of the SuMax sketch. When a switch receives an INT packet, it randomly selects a sketch column as a sketchlet, and embedded it into the packet header to send to end-host.

\begin{figure}
  \centering
  \includegraphics[width=3.5in]{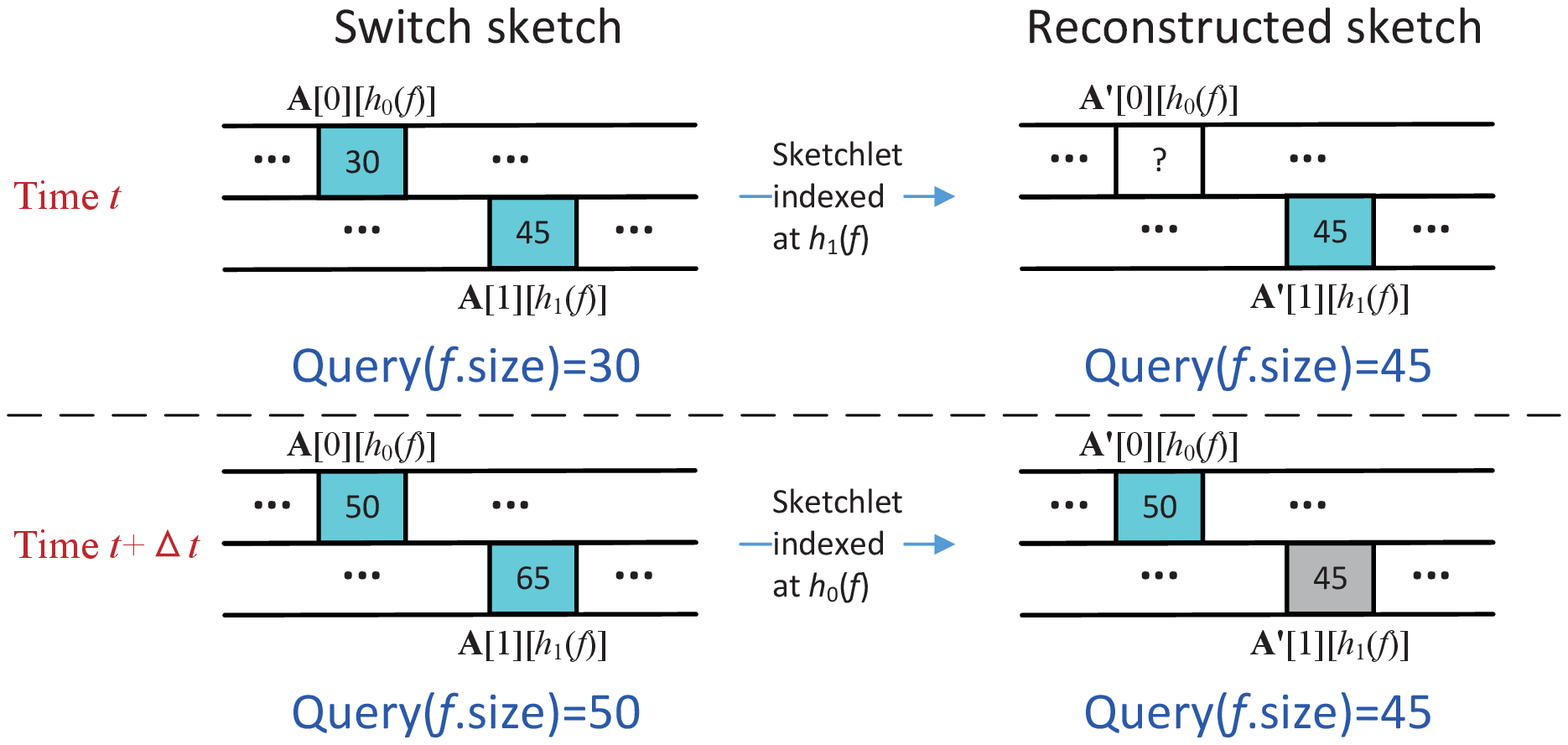}\\
  \caption{An example demonstrating the reasons behind the inaccuracies of an end-host reconstructed sketch comparing with the original switch sketch.}
  \label{Fig_Motivation}
\end{figure}

We use an example in Fig. \ref{Fig_Motivation} to demonstrate why comparing with the original switch sketch, an end-host sketch reconstructed from sketchlets is inaccurate. Suppose that the sketch is composed of $d=2$ rows, and it traces flow size in number of packets or bytes. As shown in the top figure, at time $t$, if we estimate flow $f$'s size with the switch sketch $\mathbf{A}$,  the result should be $\min\{\mathbf{A}[0][h_0(f)], \mathbf{A}[1][h_1(f)]\}=30$. However, suppose that by time $t$, only the column indexed at $h_1(f)$ has been sent to the end-host by INT, then on the reconstructed sketch $ \mathbf{A'}$, its bucket $\mathbf{A'}[0][h_0(f)]$ is invalid, as it does not contain any valid measurement data. If we estimate $f$'s size with $ \mathbf{A'}$, the result would be $\min\{\mathbf{A'}[0][h_0(f)], \mathbf{A'}[1][h_1(f)]\}=45$, which is overestimated and inaccurate, due to the invalid data in the bucket of $\mathbf{A'}[0][h_0(f)]$\footnote{According to \cite{ZhaoLightGuardian21}, the SuMax algorithm does not take an invalid bucket into the estimation when answering a query.}.

After $\Delta t$ seconds, as the flow grows, both $\mathbf{A}[0][h_0(f)]$ and $\mathbf{A}[1][h_1(f)]$ are increased by $20$. We assume that the column indexed at $h_0(f)$ has just been sent to the end-host, but the column at $h_1(f)$ has not been sent again during $[t,t+\Delta t]$, as demonstrated in the bottom figure. At time $t+\Delta t$, querying flow $f$ with the switch sketch $\mathbf{A}$ returns $50$, but querying at the reconstructed sketch $\mathbf{A'}$ returns $45$, which is underestimated and inaccurate, due to the stale data in $\mathbf{A'}[1][h_1(f)]$.

From the example, one can see that the cause of the errors in a reconstructed sketch is the invalid and stale data in its buckets. For effectively combining sketch and INT, it is essential to eliminate these errors.

%Such errors are difficult to avoid because of two reasons: First, INT packets have limited capacity to carry sketchlets to end-host, and given the huge number of sketch buckets, a bucket in the switch sketch is unlikely to be carried to the end-host right after its data gets updated. Second, there lacks effective method to select buckets that contains

\subsection{System Overview}

\begin{table}
  \centering
  \caption{Frequently used notations.}\label{Tbl_Denotation}
    \begin{tabular}{c|l}
        \hline
        \textbf{Denotation} & \textbf{Meaning} \\ \hline
        $\mathbf{A}$ & Switch sketch \\
        $\mathbf{A'} $ & End-host reconstructed sketch \\
        $\mathbf{B}$ & Bitmap \\
        $\mathbf{C}$ & Cookie \\
        $d$ & Num. of sketch/bitmap/Cookie rows \\
        $w$ & Num. of sketch/bitmap/Cookie columns \\
        $c$ & Sketch bucket size in bits \\
        $r$ & Scatter sketch's offset length in bits \\
        $b$ & Cookie cell size in bits \\
        $N$ & Num. of network flows traced by a switch sketch \\\hline
%        $s$ & Num. of bits right-shifted after a Cookie cell is selected \\
    \end{tabular}
\end{table}

%In this paper, we present \emph{DUNE}, a lightweight and accurate sketch INT network measurement system. We have  three design objectives:
%\begin{itemize}
%  \item \emph{Lightweight}: The system should be lightweight regarding network overhead by following the ``reconstructing sketch at end-host'' approach in combining sketch and INT.
%  \item \emph{Accurate}: The system should provide accurate network measurement results by reducing the inaccuracies caused by invalid and stale data in the end-host reconstructed sketch, as we have seen in Sec. \ref{Sec_Motivation}.
%  \item \emph{Practical}: The system should be practical regarding realization. In particular, the key components of the system should be designed under the constraints of the RMT programmable switches (e.g., the P4 Tofino switch).
%\end{itemize}

We present \emph{DUNE}, a lightweight and accurate sketch-INT network measurement system. DUNE is lightweight as it follows the ``reconstructing sketch at end-host'' approach. To avoid measurement errors, we make two key innovations: The first is a novel sketchlet design named \emph{scatter sketchlet}.
%, which enables a switch to select individual sketch buckets into sketchlet.
Unlike the \emph{column sketchlet} in LightGuardian \cite{ZhaoLightGuardian21} that contains an entire sketch column, a scatter sketchlet enables a switch to select individual buckets from each row of the sketch to add to sketchlet.
We further prove in theory that unless the sketch is extremely crowded, a scatter sketchlet has a higher efficiency in transferring measurement data to end-host comparing with a column sketchlet.
We present the detailed design and analysis in Sec. \ref{Sec_Scatter}.

The second innovation is a family of methods for selecting sketch buckets to sketchlets. We develop a \emph{bitmap} data structure that traces the update status of the sketch buckets, and present a bitmap-based bucket selection algorithm.
We also develop a counter array structure named \emph{Cookie} for tracing ``freshness'' of the sketch buckets, and present a Cookie-based algorithm that selects buckets containing valuable measurement data to add to sketchlets. We implement both algorithms on the P4-programmable Tofino switch, and we also propose a Cookie-based algorithm for software switches. We present the algorithms in Sec. \ref{Sec_SelectAlg}, and describe the implementations on the Barefoot Tofino switch in Sec. \ref{Sec_Implement}. Table \ref{Tbl_Denotation} lists the frequently used notations in this paper.

\section{Scatter Sketchlet}\label{Sec_Scatter}

\begin{figure}
  \centering
  \includegraphics[width=3.2in]{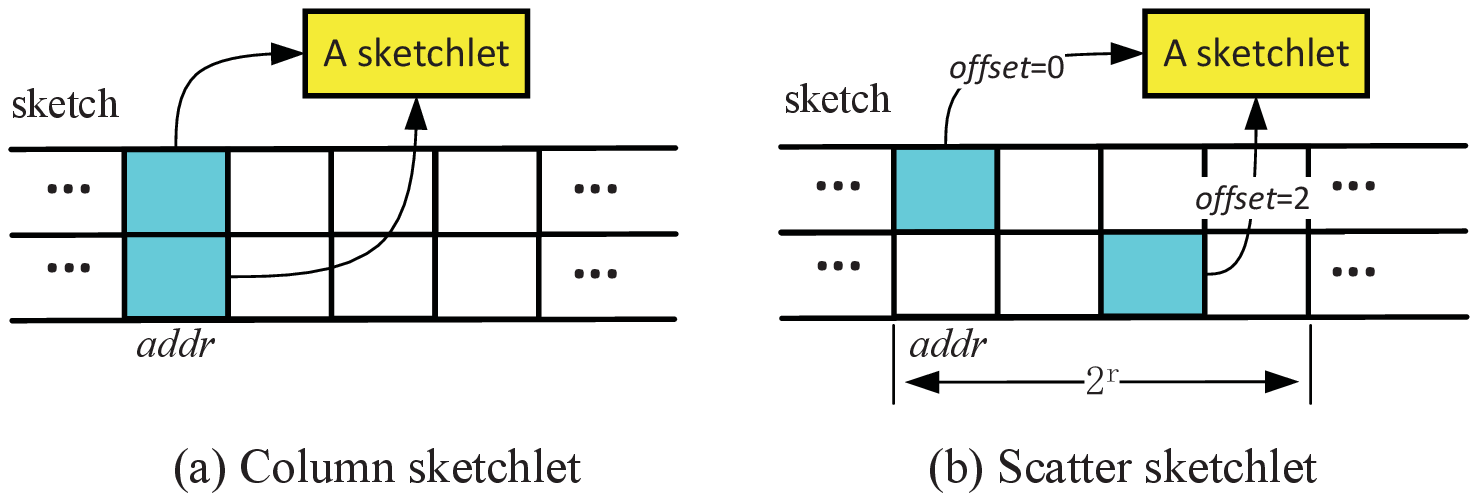}\\
  \caption{A comparison between column sketchlet and scatter sketchlet.}
  \label{Fig_Sketchlet}
\end{figure}

\subsection{Sketchlet Design}

Before presenting our sketchlet design, we first describe how a sketch is realized and how a sketchlet is formed on a Tofino switch. A Tofino switch processes network packets with a pipeline, which is composed of a series of match-action unit (MAU) stages. Each MAU stage has a stage-local memory, and stateful elements such as sketch buckets are stored as \emph{registers} in the memory. A register can be accessed at most once by a packet in its pipeline pass, moreover, a register access is limited to one simple read-update-write operation that must be realized in a small piece of code called \emph{register action}.
%Moreover, a register action is limited to one simple read-update-write operation, so as to ensure that the register access is atomic.
In a sketch-based measurement system,
%A memory can be accessed at most once per pipeline pass, that means a packet can perform the read-update-write operation on each memory only once. Moreover, a memory access can read/write at most $64$ consecutive bytes.
%In sketch-based measurement systems,
a $d\times w$ sketch is generally realized as $d$ registers, each contains a row of $w$ buckets, and different registers are placed in different stage-local memories on a Tofino switch. When an INT packet enters into the pipeline, it sequentially accesses the registers, and retrieves the bucket at the specified column index from each register to form a column sketchlet.
%Table \ref{Tbl_Denotation} lists the frequently used symbols in this paper.
%from each register, the bucket at the specified column index is read and written to the sketchlet in the packet header.
%As a result, one register is accessed only once.

We propose a novel sketchlet design named \emph{scatter sketchlet}. As shown in Fig. \ref{Fig_Sketchlet}, for a $d\times w$ switch sketch $\mathbf{A}$, a scatter sketchlet contains $d$ buckets, one from each sketch row. A scatter sketchlet is addressed as $(addr,offset[1\cdots d] )$, where $addr\in \{1,\cdots, w\}$ is a column index, and $offset[i]$ is an $r$-bit offset indicating the distance between the column index of the bucket in the $i^{th}$ row and $addr$. For example, $(addr, offset[i])$ points to the sketch bucket $\mathbf{A}[i][(addr+offset[i])]$.
%An offset is $r$ bits long, so an offset's value is between $0$ and $2^r-1$.
Note that the column sketchlet can be viewed as a special case of the scatter sketchlet with $r=0$.

Scatter sketchlet doesn't violate the Tofino switch's memory access restriction, as one register, which implements a sketch row, is still accessed at most once. The only difference is that we allow a bucket to be selected from a range of $[addr,(addr+2^r-1)]$ rather than at a fixed $addr$. We analyze the advantage of the design in the following subsection.

\subsection{Bit Efficiency}

%\begin{figure}
%  \centering
%  \includegraphics[width=2.2in]{../Experiment/Numerical/1.eps}\\
%  \caption{Probabilities of all valid buckets of column and scatter sketchlets from a sketch of $d=2$ and $w=2^{16}$ tracing various number of flows $N$, the scatter sketchlet uses various offset length $r$.}
%  \label{Fig_ValidProb}
%\end{figure}

The design of the scatter sketchlet allows its sketch bucket to be selected from a range $[addr, (addr+2^r-1)]$, which greatly reduces the chance that an invalid bucket is selected. Formally, consider a $w\times d$ sketch tracing a total number of $N$ network flows, a bucket is invalid only when none of the $N$ flows is hashed to it, which happens at a probability of $\left( 1-\frac{1}{w}\right) ^{N}\approx e^{-\frac{N}{w}}$. The probability that there exists at least one valid bucket in the range $[addr, (addr+2^r-1)]$ is $\left( 1-e^{-\frac{N}{w}\times 2^{r}}\right)$, and among the $d$ buckets in a scatter sketchlet, averagely $\left( 1-e^{-\frac{N}{w}\times 2^{r}}\right)\times d$ of them contain valid measurement data.

%the probability that all the buckets are valid is
%\begin{equation}
%P=\left( 1-e^{-\frac{N}{w}\times 2^{r}}\right) ^{d} \label{Equ_Valid}
%\end{equation}
%If we let $r=0$, Eq. (\ref{Equ_Valid}) is indeed the probability of a column sketchlet to have all valid buckets.
%We numerically examine the probabilities of two different sketchlet designs in Fig. \ref{Fig_ValidProb}. From the figure one can see that with a small $r$ (i.e., $r\leq6$), a scatter sketchlet has a much higher chance to have all valid buckets than a column sketchlet.

%\subsubsection{Bit efficiency}

Besides the buckets, a sketchlet also needs to carry the addresses of its contained buckets for the end-host to reconstruct the sketch. For a column sketchlet, its address is simply the $\log_{2} w$-bit column index, but for a scatter sketchlet, in addition to the $\log_{2}w$-bit $addr$, the address also contains the $d$ offsets, and the total size is $(\log_{2} w+d\times r)$ bits.

We define \emph{bit efficiency}, which is the ratio between the bits of the valid measurement data in a sketchlet and the total sketchlet size, to measure the efficiency of a sketchlet in transferring measurement data to end-host. From the above analysis, it is easy to see that a scatter sketchlet's bit efficiency is
\begin{equation}
E=\frac{\left( 1-e^{-\frac{N}{w}\times 2^{r}}\right)\times d\times c}{d\times c+\log _{2}w+d\times r} \label{Equ_Efficiency}
\end{equation}
where $c$ is the size of a sketch bucket in bits. Note that a column sketchlet's bit efficiency can be obtained by applying $r=0$ to Eq. (\ref{Equ_Efficiency}). For comparing the bit efficiencies of the two sketchlet designs, we have the following result.

\begin{restatable}{theorem}{Upper}\label{Thm_Upper}
 As long as the number of the network flows $N$ traced by a sketch satisfies
\begin{equation}
N< w\ln \left( \frac{d\times c+\log _{2}w}{d\times r}\right) \label{Equ_Upper}
\end{equation}
a scatter sketchlet with $r\geq 1$ achieves a higher bit efficiency than a column sketchlet.
\end{restatable}
\begin{proof}
 A scatter sketchlet achieves a higher bit efficiency than a column sketchlet only when
\[
\frac{\left( 1-e^{-\frac{N}{w}\times 2^{r}}\right) \times d\times c}{d\times
c+\log _{2}w+d\times r} > \frac{\left( 1-e^{-\frac{N}{w}}\right) \times
d\times c}{d\times c+\log _{2}w}
\]%
which is equivalent to
\[
\frac{e^{-\frac{N}{w}}-e^{-\frac{N}{w}\times 2^{r}}}{1-e^{-\frac{N}{w}}} >
\frac{d\times r}{d\times c+\log _{2}w}
\]%
Note that for $r\geq1$,
\[
\frac{e^{-\frac{N}{w}}-e^{-\frac{N}{w}\times 2^{r}}}{1-e^{-\frac{N}{w}}}\geq e^{-%
\frac{N}{w}}
\]
When $N<w\ln \left( \frac{d\times c+\log _{2}w}{d\times r}\right) $, we have $e^{-%
\frac{N}{w}}>\frac{d\times r}{d\times c+\log _{2}w}$, which leads to
\[
\frac{e^{-\frac{N}{w}}-e^{-\frac{N}{w}\times 2^{r}}}{1-e^{-\frac{N}{w}}}>%
\frac{d\times r}{d\times c+\log _{2}w}
\]
 \end{proof}

\begin{figure}[tbp]
  \centering
  \subfigure[Upper bound number of flows]{\includegraphics[height=1.30in]{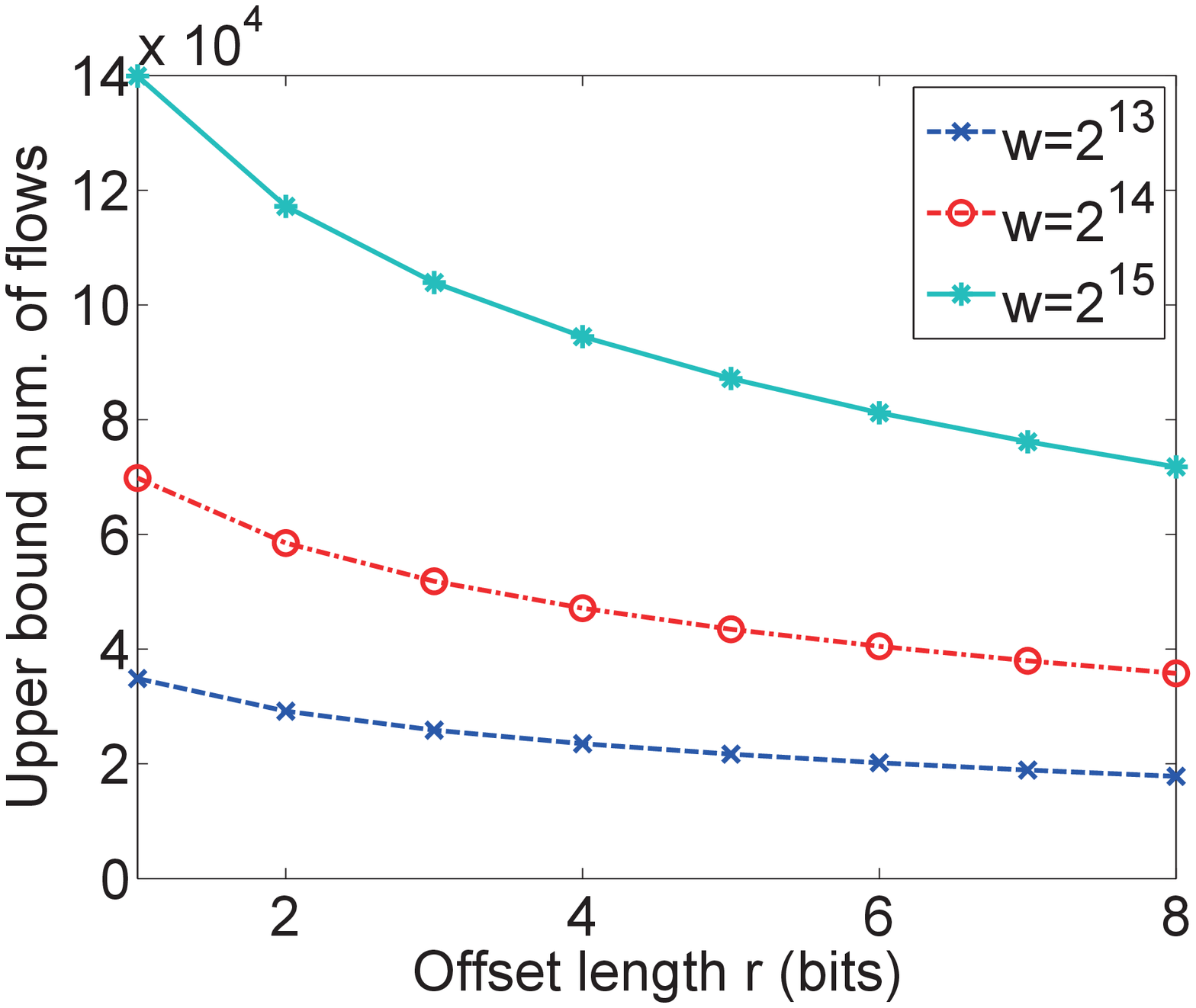}}
  \subfigure[Hash collision probability]{\includegraphics[height=1.30in]{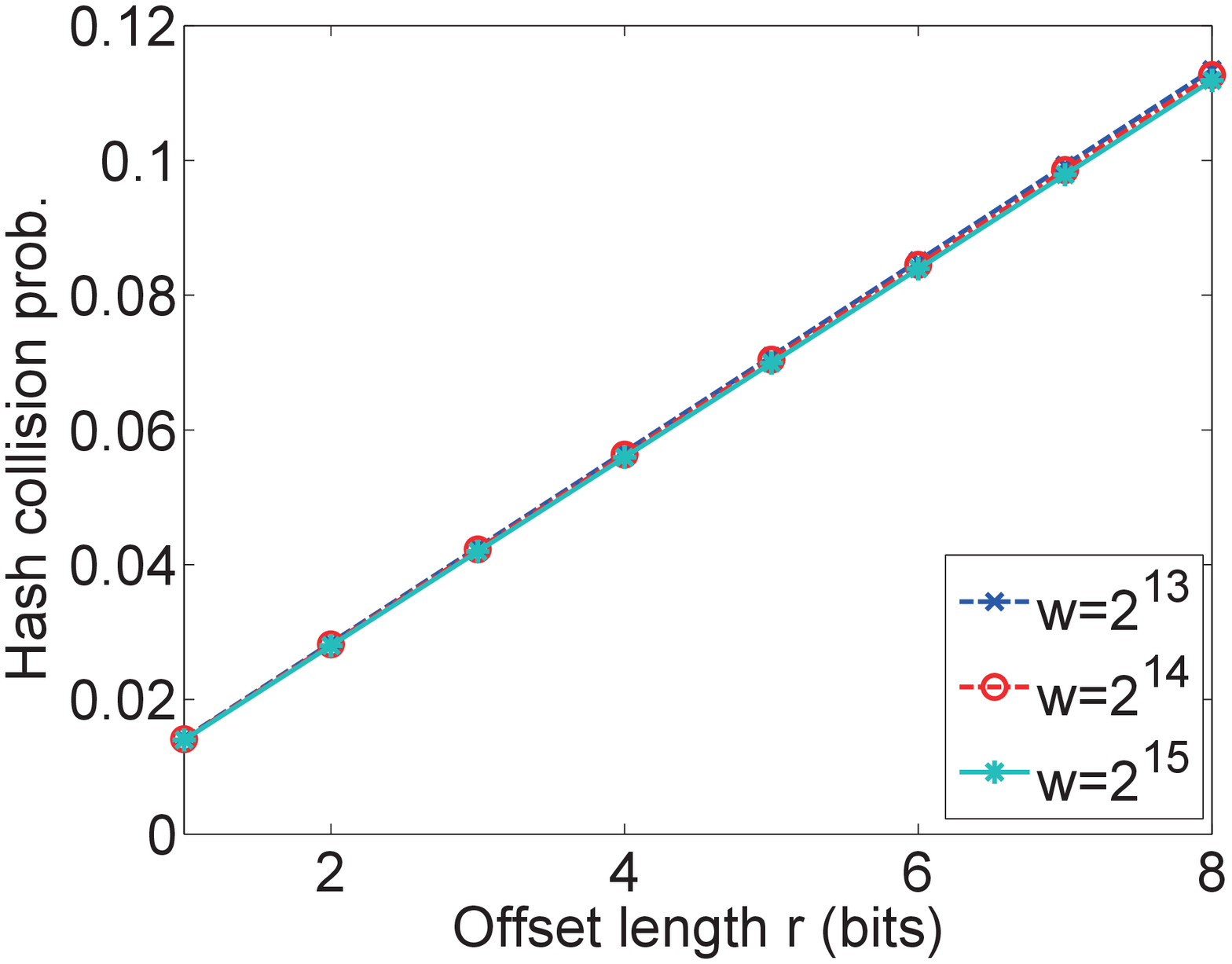}}
  \caption{(a) Upper bound number of flows that a scatter sketchlet can achieve a higher bit efficiency than a column sketch, and (b) hash collision probability of a sketch when tracing the upper bound number of flows. The sketch has $d=2$ rows and various number $w$ of columns, bucket size is $c=64$ bits, and the offset length $r$ of the scatter sketchlet varies from $1$ to $8$ bits.}
  \label{Fig_BitEfficiency}
\end{figure}

In Fig. \ref{Fig_BitEfficiency}(a), we plot the upper bound flow numbers in Eq. (\ref{Equ_Upper}) under various sketch sizes and scatter sketchlet offset lengths, and Fig. \ref{Fig_BitEfficiency}(b) presents the sketch's hash collision probabilities when tracing the upper bound numbers of flows.
We can see that a scatter sketchlet achieves a higher bit efficiency, even though the number of the flows traced by the sketch far exceeds the sketch's total number of buckets. In addition, when tracing the upper bound number of flows, the sketch is indeed overly crowded and has very high hash collision probabilities, therefore can no longer provide accurate measurement results.
In other words, \emph{comparing with the column sketchlet, our proposed scatter sketchlet is more efficient in transferring measurement data to end-host.}

%We first analyze the column sketchlet used in LightGuardian \cite{ZhaoLightGuardian21}. Suppose the switch sketch traces $N$ network flows, a bucket is invalid only when none of the flows is hashed to it, which happens at a probability of $\left( 1-\frac{1}{w}\right) ^{N}\thickapprox e^{-\frac{N}{w}}$. A column sketchlet only when all its buckets are valid, which happens at $\left( 1-e^{-\frac{N}{w}}\right) ^{d}$.

%Unlike the column sketchlet as in \cite{ZhaoLightGuardian21}, which may carries invalid buckets to end-host.
%With a scatter sketchlet, a switch can avoid packing an invalid bucket to the sketchlet.

%, and for all the $d$ buckets, the minimum offset is $0$.

\section{Bucket Selecting Algorithm}\label{Sec_SelectAlg}

Our proposed scatter sketchlet allows a switch to pick a sketch bucket in a range of $[addr, (addr+2^r-1)]$ to add to sketchlet, however, how to select the bucket that contains valuable measurement data is still unknown. In this section, we present algorithms for selecting sketch buckets.

\subsection{Bitmap Algorithm}

The first algorithm we propose is called \emph{bitmap algorithm}.
%, which avoid a switch to select a sketch bucket if its measurement data is invalid or has not been changed since the last time the bucket is selected.
As its name suggests, the algorithm maintains within programmable switch a bitmap $\mathbf{B}$, which has same logical structure as the sketch with $d$ rows and $w$ columns of bits. Initially, all bits are set as $0$.

\begin{algorithm}\label{Alg_Bitmap}
\SetKwFunction{algo}{}
\SetKwInOut{Input}{Input}\SetKwInOut{Output}{Output}
\small{
\Input{A packet of flow $f$}
    \If{{\upshape $f$ is an ordinary flow}}{
        \For{$i=1\cdots d$}{
            Update $\mathbf{A}[i][h_i(f)]$\;
            $\mathbf{B}[i][h_i(f)]\leftarrow 1$\;
        }
    }
    \If{{\upshape $f$ is the INT flow}}{
        Randomly select $addr$ from $\{1,\cdots ,w\}$\;
        \For{$i=1\cdots d$}{
            \For{$j=0\cdots 2^r-1$}{
                \If{$\mathbf{B}[i][addr+j]==1$}{
                    break\;
                }
            }
            $\mathbf{B}[i][addr+j]\leftarrow 0$\;
            Add $\mathbf{A}[i][addr+j]$ to scatter sketchlet\;
        }
    }
}
\caption{\emph{Bitmap} algorithm}
\end{algorithm}

As presented in Algorithm \ref{Alg_Bitmap}, a bitmap $\mathbf{B}$ in a switch is updated on two events:
\begin{itemize}
  \item When the switch receives a packet of flow $f$, in addition to update the mapped sketch buckets $\mathbf{A}[i][h_i(f)]$, the switch also sets all the bits at the same positions of the bitmap as $1$, i.e., $\mathbf{B}[i][h_i(f)]=1$, for $i=1,\cdots,d$ (line 1-4).
  \item When receiving an INT flow packet, the switch randomly selects $addr$ from $\{1,\cdots, w\}$ (line 6), and from each row of the bitmap, it finds the first bit in the range of $[addr, (addr+2^r-1)]$ whose value is $1$, adds the corresponding sketch bucket to the sketchlet, and clears the bit to $0$ (line 7-12).
\end{itemize}

%all the bits in the range of $[addr, addr+2^r-1]$, and selects a sketch bucket $\mathbf{A}[i][j]$ whose corresponding bit $\mathbf{B}[i][j]$ is largest among them, to add into the sketchlet (line 7-8). $\mathbf{B}[i][j]$ cleared to 0 afterwards (line 9).

%As presented in Algorithm \ref{Alg_Bitmap}, the algorithm updates the bitmap $\mathbf{B}$ on two events: 1) When the switch receives a packet of flow $f$, in addition to update the mapped sketch buckets $\mathbf{A}[i][h_i(f)]$, the switch also sets all the bits at the same positions of the bitmap as $1$, i.e., $\mathbf{B}[i][h_i(f)]=1$, for $i=1,\cdots,d$ (line 2-3). 2) When receiving an INT flow packet, the switch randomly selects $addr$ from $[1\cdots w]$ (line 5), and from each row of $\mathbf{B}$, it reads all the bits in the range of $[addr, addr+2^r-1]$, and selects a sketch bucket $\mathbf{A}[i][j]$ whose corresponding bit $\mathbf{B}[i][j]$ is largest among them, to add into the sketchlet (line 7-8). $\mathbf{B}[i][j]$ cleared to 0 afterwards (line 9).

The bitmap algorithm achieves two objectives: First, it avoids adding invalid bucket to sketchlet, as an invalid bucket's corresponding bit in the bitmap is always $0$; Second, it will not select a bucket if it has not been updated since the last time it is selected into a sketchlet.

We elaborate how to place the bitmap and implement Algorithm \ref{Alg_Bitmap} in a Tofino switch in Sec. \ref{Sec_Implement}.

%We aim to realize the bitmap algorithm on Tofino switch. However, although the bitmap $\mathbf{B}$ has a same logical structure as the sketch $\mathbf{A}$, unfortunately, it can not be implemented using the same method as in implementing the sketch. We discuss the challenges and elaborate our solutions in Sec. \ref{Sec_Implement}.

%Note that Algorithm \ref{Alg_Bitmap} is carefully designed under the constraints of the P4 Tofino switch, and we discuss its Tofino implementation in details in Sec. \ref{Sec_Implement}.

%Algorithm \ref{Alg_Bitmap} can be realized with P4 on Tofino switches. In particular, each row of the bitmap is realized as a per-stage register.
%, and $d$ additional registers, each contains $w$ bits, are required.
%An INT packet accesses $2\times d$ registers per pipeline pass: $d$ register accesses to $d$ sketch rows and $d$ accesses to $d$ bitmap rows. Note that since a register access can read/write up to $64$ consecutive bytes, as long as $r\leq 9$, it is possible for the switch to retrieve all the $2^r$ bits in the range $[addr, addr+2^r-1]$ in one read. In our implementation, we set $d=2$ and $r=6$, which requires $4$ registers, and does not violate the memory access rule of the Tofino switch.

\subsection{Cookie Algorithm} \label{Sec_CookieAlg}

Studies show that rate distribution of real-world network flows is highly skewed, and under such a distribution, one flow may grow much faster than another \cite{BensonDCTraffic10}\cite{ChenDCTraffic11}. To cope with such a skewness, we propose another algorithm namely \emph{Cookie algorithm}, and present it in Algorithm \ref{Alg_Cookie}.
%which uses a counter instead of a bit to trace the status of a sketch bucket.
The algorithm maintains in programmable switch a counter array named \emph{Cookie}, which has same logical structure as the sketch with $d$ rows and $w$ columns of cells. Each Cookie cell is a $c$-bit counter, and all the counters are initialized as $0$.

\begin{algorithm}\label{Alg_Cookie}
\SetKwFunction{algo}{}
\SetKwInOut{Input}{Input}\SetKwInOut{Output}{Output}
\small{
\Input{A packet of flow $f$}
    \If{{\upshape $f$ is an ordinary flow}}{
        \For{$i=1\cdots d$}{
            Update $\mathbf{A}[i][h_i(f)]$\;
            $\mathbf{C}[i][h_i(f)] \leftarrow \mathbf{C}[i][h_i(f)]+1$\;
        }
    }
    \If{{\upshape $f$ is the INT flow}}{
        $PktCnt++$; Randomly select $addr$ from $\{1,\cdots ,w\}$\;
        \For{$i=1\cdots d$}{
            \For{$j=0\cdots 2^r-1$}{
                \If{$\mathbf{C}[i][addr+j] \geq (2^h-1)$}{
                    $CellCnt++$; break\;
                }
            }
            $\mathbf{C}[i][addr+j]\leftarrow \frac{\mathbf{C}[i][addr+j]}{2^s}$\;
            Add $\mathbf{A}[i][addr+j]$ to scatter sketchlet\;
        }
    }
}
\caption{\emph{Cookie} algorithm}
\end{algorithm}

Similar to the bitmap algorithm, the Cookie algorithm updates the Cookie structure $\mathbf{C}$ on two events:
\begin{itemize}
  \item When the switch receives a packet of flow $f$, in addition to update the sketch buckets $\mathbf{A}[i][h_i(f)]$, the switch also increments the counters $\mathbf{C}[i][h_i(f)]$ by 1, for $i = 1,\cdots, d$ (line 1-4).
  \item When receiving an INT flow packet, the switch randomly selects $addr$ from $\{1,\cdots, w\}$ (line 6), and for each row in the Cookie, it compares each cell in the range $[addr, (addr+2^r-1)]$ against a threshold $(2^h-1)$ (where $h\leq b$) one by one. For the first cell that is not smaller than $(2^h-1)$, its corresponding sketch bucket is selected into the sketchlet, and the cell's value is reduced by being right-shifted $s$ bits (line 7-12).
\end{itemize}
%2) When a sketch bucket $\mathbf{A}[i][j]$ is selected into a sketchlet, the switch right-shifts the counter $\mathbf{C}[i][j]$ $s$ bits, i.e., $\mathbf{C}[i][j]=\frac{\mathbf{C}[i][j]}{2^s}$, where $1\leq s\leq b$.
%The threshold $\omega$ increases by a factor of $\rho$ each time a larger cell value can be found  in the range  $[addr, addr+2^r-1]$, and decreases by factor $\rho$ otherwise (line 13-16).

The switch maintains two counters in the pipeline, $CellCnt$ and $PktCnt$. $CellCnt$ records the number of the sketch buckets selected into sketchlets (line 10), and $PktCnt$ is the number of the INT packets it has received (line 6). Periodically, the switch computes a ratio $\frac{CellCnt}{d\times PktCnt}$: if the ratio is below a threshold $\alpha$, the switch decreases the parameter $h$ by one as $h=h-1$, which means that the algorithm is less selective in adding buckets to sketchlets; and if the ratio is larger than another threshold $\beta$, the algorithm behaves more selective with $h=h+1$.

%it reads all the cells in the range of $[addr, addr+2^r-1]$, selects a sketch bucket $\mathbf{A}[i][j]$ whose corresponding Cookie cell $\mathbf{C}[i][j]$ is the largest among the range to add to the sketchlet, and reduce $\mathbf{C}[i][j]$ by right-shifting its value by $s$ bits afterwards (line 7-9). %The complete Cookie algorithm is presented Algorithm \ref{Alg_Cookie}, and the algorithm has the following property.

For the Cookie algorithm, we have the following result.
\begin{restatable}{theorem}{Proportion} \label{Thm_Proportion}
The frequency of a sketch bucket being selected into sketchlets is statistically proportional to the bucket's update frequency.
\end{restatable}
\begin{proof}
Consider a set of $n$ sketch buckets, whose update frequencies are $f_{1},\cdots ,f_{n}$. Suppose averagely in every $t$ seconds, a sketch bucket is selected into a sketchlet, and its corresponding cell counter value is reduced to $\frac{1}{2^{s}}$. Let $K$ be the averaged counter value before a Cookie cell is reduced, then to reach an equilibrium, we have
\[
t\times \sum_{i=1}^{n}f_{i}=K\left( 1-\frac{1}{2^{s}}\right)
\]%

For bucket $i$, averagely each time its associated Cookie cell is increased from $K \times \frac{1}{2^{s}}$ to $K$, which takes an interval of
\[
\frac{K\times (1-\frac{1}{2^{s}})}{f_{i}}
\]
it is selected into a sketchlet, therefore bucket $i$'s selection frequency is
\[
\frac{f_{i}}{K\times (1-\frac{1}{2^{s}})}
\]%
which is proportional to its update frequency $f_{i}$.

\end{proof}

Theorem \ref{Thm_Proportion} indicates that the Cookie algorithm can accurately identify the frequently updated sketch buckets, and add them to sketchlets. Obviously, this is a desired property in network measurement, especially when network flows follow a highly skewed rate distribution.

We elaborate how to place the Cookie structure and implement Algorithm \ref{Alg_Cookie} in a Tofino switch in Sec. \ref{Sec_Implement}.

%As the bitmap algorithm, Algorithm \ref{Alg_Cookie} is also designed by taking the P4 Tofino switch constraints into consideration. We discuss its implementation on Tofino switch in Sec. \ref{Sec_Implement}.

%Algorithm \ref{Alg_Cookie} is realized on Tofino switches in a similar way as the bitmap algorithm. In particular, each row of the Cookie is realized as a per-stage register containing $w$ cells, and a packet accesses $2 \times d$ registers per pipeline pass: $d$ accesses to $d$ sketch rows and $d$ accesses to $d$ Cookie rows. Note that as long as $2^r\times b \leq 512$, a switch can read all the $2^r$ Cookie cells from the range $[addr, addr + 2^r -1]$ without violating Tofino's memory access rule. In our implementation, we set the counter size as $b=8$ bits and let $r=6$.

\subsection{Software Switch Algorithm}

Both Algorithm \ref{Alg_Bitmap} and Algorithm \ref{Alg_Cookie} are designed for RMT programmable switches.
%search sketch bucket in a limited range of $[addr, addr + 2^r -1]$ starting from a randomly selected $addr$, due to the memory access restriction of the Tofino switch.
In the following, we present an algorithm for software switch in Algorithm \ref{Alg_Software}. We propose the algorithm for two reasons: First, the algorithm can work on software switches such as OVS \cite{PfaffOVS15}; Second, by exhaustively searching sketch buckets that contain valuable measurement data, the algorithm can provide a benchmark for comparison.

\begin{algorithm}\label{Alg_Software}
\SetKwFunction{algo}{}
\SetKwInOut{Input}{Input}\SetKwInOut{Output}{Output}
\small{
\textbf{Repeat}:\\
    $addr=NULL$, $offset[1\cdots d]=NULL$\;
    \For {$i=1 \cdots w$}{
        \If{$(addr!=NULL) \&\& (offset[1\cdots d]!= NULL)$}{
            Add $(addr, offset[1\cdots d])$ to FIFO\;
            \For {$j=1$ to $d$}{
                $\mathbf{C}[j][addr+offset[j]]=\frac{C[j][addr+offset[j]]}{2^s}$\;
            }
            $addr=NULL$; $offset[1\cdots d]=NULL$\;
        }
        \ElseIf{$(addr!=NULL) \&\& (i-addr\geq 2^r)$}{
            Find $o=\min\{offset[1], \cdots, offset[d]\}$\;
            $addr=addr+o$\;
            \For{$j=1$ to $d$}{
                $offset[j]=offset[j]-o$\;
                \If{$offset[j]<0$}{
                    $offset[j]=NULL$\;
                }
            }
        }
        \For{$j=1$ to $d$}{
            \If{$\mathbf{C}[i][j] \geq (2^h-1)$}{
                \If{$addr==NULL$}{
                    $addr=i$; $offset[j]=i-addr$\;
                }
                \ElseIf{$(i<(addr+2^r)) \&\& (offset[j]==NULL)$}{
                    $offset[j]=i-addr$\;
                }
            }
        }
    }
}
\caption{Software switch algorithm}
\end{algorithm}

%As presented in Algorithm \ref{Alg_Software}, the software switch algorithm employs a Cookie structure as in Sec. \ref{Sec_CookieAlg}, and it maintains an \emph{FIFO queue} to keep the addresses of the selected candidate sketchlets.
%Unlike the bitmap and the Cookie algorithms in which bucket selection is driven by packet reception,
Algorithm \ref{Alg_Software} employs a Cookie structure as in Sec. \ref{Sec_CookieAlg} for tracing ``freshness'' of the data in sketch buckets. But unlike the bitmap and Cookie algorithms in which bucket selection is driven by packet reception, the algorithm proactively scans the Cookie structure (line 1-3), and uses an \emph{FIFO queue} to keep the addresses of the candidate sketchlets.
More specifically, the algorithm keeps the addresses of the sketch buckets it currently selects in a tuple $(addr, offset[1\cdots d])$.
From each row of the Cookie, if the algorithm has found a Cookie cell in the range of $[addr, (addr+2^r-1)]$ whose value is no smaller than a threshold $(2^h-1)$ (where $1\leq h\leq b$), the algorithm records the cell's address in the tuple (line 16-21). When all the $d$ sketch buckets have been successfully selected, the algorithm reduces their Cookie cells by right-shifting the counters $s$ bits, and adds the address tuple to the FIFO queue (line 4-8).
%If such an address tuple can be found, the algorithm right-shift the counters of the Cookie cells, and adds the tuple to the FIFO queue (line 3-6).
If in some of the $d$ rows, no bucket can be found in the range $[addr, (addr+2^r-1)]$, the algorithm moves forward $addr$, clears some offsets whose positions are behind the new $addr$ so as to unselect the corresponding sketch buckets, and continue to search from the new $addr$ in the Cookie structure  (line 9-15).

Each time a switch receives an INT flow packet, it removes an address tuple from the FIFO queue, and adds the corresponding sketch buckets pointed by the address tuple to the sketchlet. The software switch monitors the size of the FIFO queue, and compares it with two thresholds, $\phi$ and $\varphi$. When the queue size is smaller than $\phi$, the switch decreases the algorithm parameter $h$ by one as $h=h-1$, and if the queue size is larger than $\varphi$, the switch increases $h=h+1$ to pick the sketch buckets more selectively.

Clearly, Theorem \ref{Thm_Proportion} also applies to Algorithm \ref{Alg_Software}, as the algorithm updates the Cookie cells in a same way as in Algorithm \ref{Alg_Cookie}.
%By dynamically adjusting $s$, the algorithm can control standard for selecting sketch buckets into sketchlets.
%scans the entire Cookie from the first column to the last one. When the algorithm has found a Cookie cell whose counter is larger than a threshold of $2^s$ ($1\leq s\leq b$), if the cell's address is not considered

\section{Prototype Implementation}\label{Sec_Implement}

We have implemented a prototype of the DUNE system, and in particular, we have realized the bitmap and the Cookie algorithms on Edgecore Wedge 100BF Tofino-based programmable switches.

For implementing a DUNE switch, two components need to be realized: 1) the sketch structure and the action to access the sketch buckets; 2) the bitmap/Cookie and the action to access the bits/Cookie cells.
For realizing the switch sketch, we follows the method in LightGuardian \cite{ZhaoLightGuardian21} and implement a SuMax sketch. The sketch has $d=2$ rows and $w=2^{15}$ columns of buckets, the size of a sketch bucket is $c=64$ bits, and a sketch row is implemented as a \SI{256}{\kilo\byte}-register. We realize the operation for updating and retrieving a sketch bucket in one single register action.

%An intuitive solution to use the same method to implement a bitmap or a Cookie structure on a Tofino switch. Unfortunately, this is infeasible because of the following reason.

Although having same logical structure, however, we can not use the same method to implement a bitmap or a Cookie, because of the following reason: Recall that in both Algorithm \ref{Alg_Bitmap} and Algorithm \ref{Alg_Cookie}, a sketch bucket is selected from a range $[addr, (addr+2^r-1)]$, and in the worst case, as many as $2^r$ bits in the bitmap or cells in the Cookie need to be inspected. If a row of bitmap/Cookie is implemented as one register, under the Tofino switch's register access rule, the operations for inspecting $2^r$ consecutive bits or Cookie cells must be realized in one single register action. Unfortunately, the current P4 Tofino switch only allows simple operations in a register action, and it is prohibitive to
%implemented the complex logics as in Algorithm \ref{Alg_Bitmap} and Algorithm \ref{Alg_Cookie} regarding
inspect $2^r$ consecutive bits or Cookie cells within one single register action.

%the P4 Tofino switch does not support to implement above operations in one single register action.

\begin{figure}
  \centering
  \includegraphics[width=3.2in]{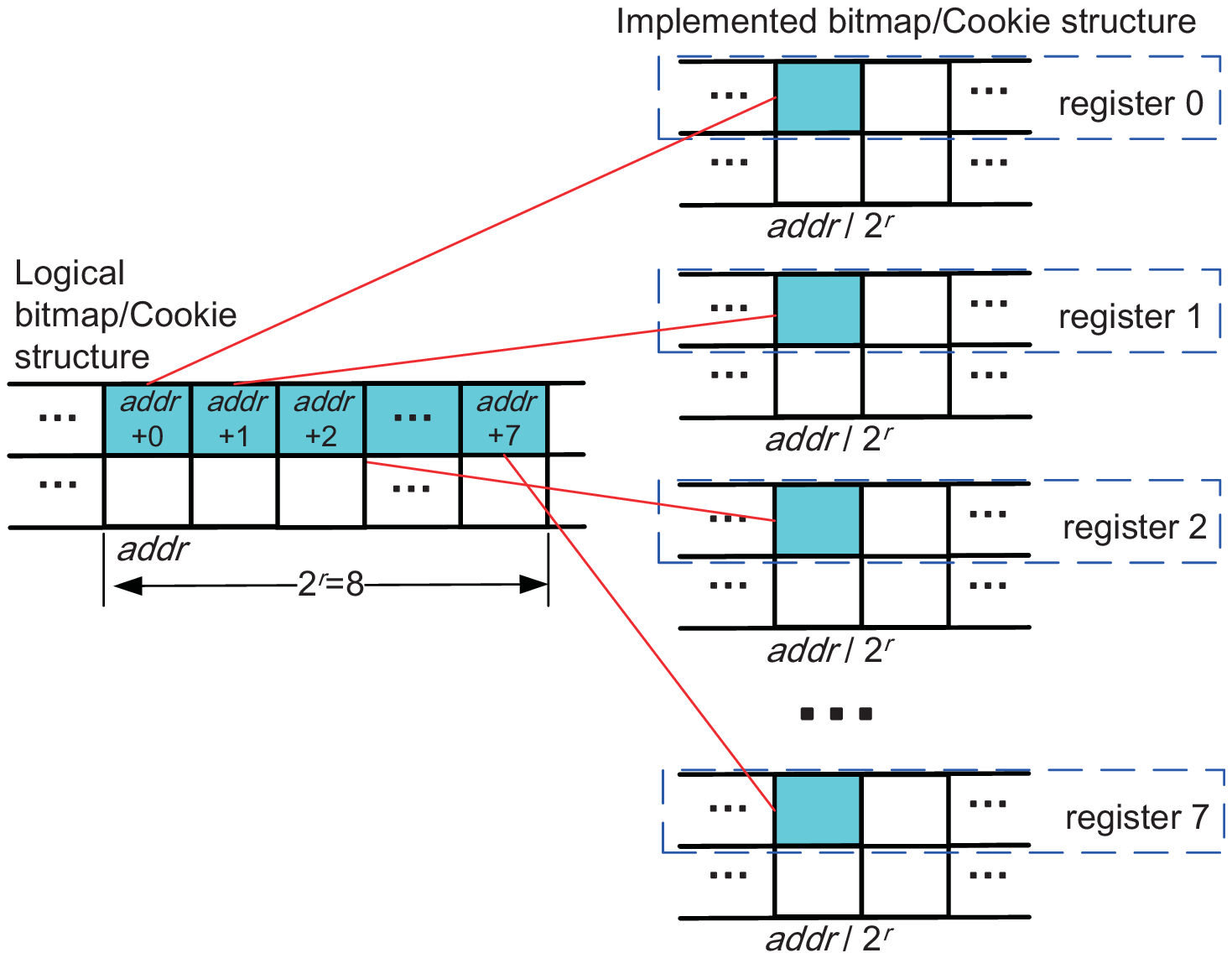}\\
  \caption{Implementation of bitmap/Cookie structure on Tofino switch.}
  \label{Fig_Implementation}
\end{figure}

%The logics of an atomic register access is implemented in a piece of code namely \emph{register action} in P4. However, only a few basic operations such as add/minus, value assignment, and value comparison are supported in a register action. Although, these basic operations are sufficient in a sketch bucket access, they are not enough to realize the $\arg\max$ operation as in Algorithm \ref{Alg_Bitmap} and Algorithm \ref{Alg_Cookie}.

To overcome the problem, in our implementation, we realize a bitmap/Cookie row with $2^r$ registers, where each register contains $\frac{w}{2^r}$ bits/Cookie cells. As shown in Fig. \ref{Fig_Implementation}, for searching in $2^r$ consecutive bits/Cookie cells in $[addr, (addr+2^r-1)]$, we actually access the bits/Cookie cells indexed at $\frac{addr}{2^r}$ in all the $2^r$ registers one by one\footnote{In our implementation, $addr$ is randomly selected as multiples of $2^r$.}. For a bit/Cookie cell at $\frac{addr}{2^r}$ in each register, we check (and update) its value with one single register action.
If a bit/Cookie cell in the $j^{th}$ register is selected, we add the bucket indexed at $(addr+j)$ from the sketch to the sketchlet.

%We describe the implementation of the bitmap/Cookie structure in the following.

%a Cookie row with $2^r$ registers, each containing $w/2^r$ cells. As shown in Fig. XX, for searching the largest cell in a range $[addr,addr+2^r-1]$, we actually read the cell indexed at $addr$ from each register, and find the largest cell value among the $2^r$ registers. For example in Fig. XX, if the cell from the $\mathbf{C}_k$ is found as the largest after comparison, it is equivalent that the $\arg\max$ operation on the range $[addr,addr+2^r-1$ returns $addr+k$.
Due to the limitations of the Tofino switch, we set $r=3$ in our implementation, therefore use $8$ registers to implement a row of bitmap or Cookie.
%We use polynomials of \texttt{crc16} as the hash functions.
We have shared the our P4 code to the community\footnote{https://github.com/DuneHPCC724/Dune}.

%Another issue is to

%For enabling the $\arg\max$ operation in sketch bucket selection, we realize the $d\times w$ bitmap or Cookie as the following.

%For each bitmap/Cookie row, its cells are stored in $2^r$ registers, each contains $w/2^r$ cells. As demonstrated in Fig. XX, when comparing the cells in

%The bitmap and the Cookie algorithms allow a tradeoff between switch memory usage and measurement accuracy. The Cookie algorithm, which places a Cookie structure of $d\times w\times b$ bits within a switch, can achieve a higher measurement accuracy than the bitmap algorithm, which only uses $d\times w$ bits. We evaluate the two algorithms in Sec. \ref{Sec_Eval}.

\section{Evaluation}\label{Sec_Eval}

\begin{figure*}[tbp]
  \centering
  \subfigure[Cardinality]{\includegraphics[height=1.33in]{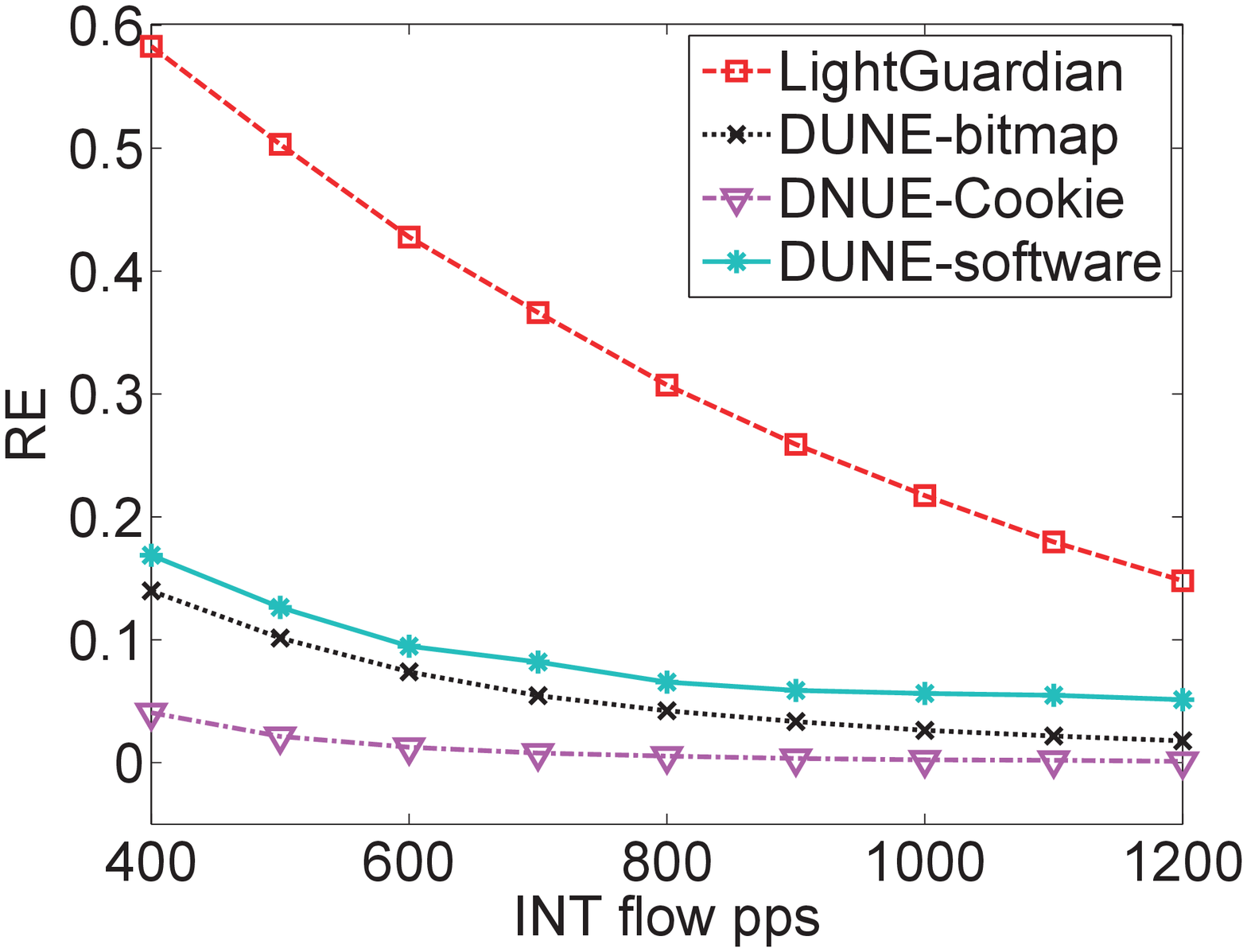}}
  \subfigure[Heavy hitter]{\includegraphics[height=1.33in]{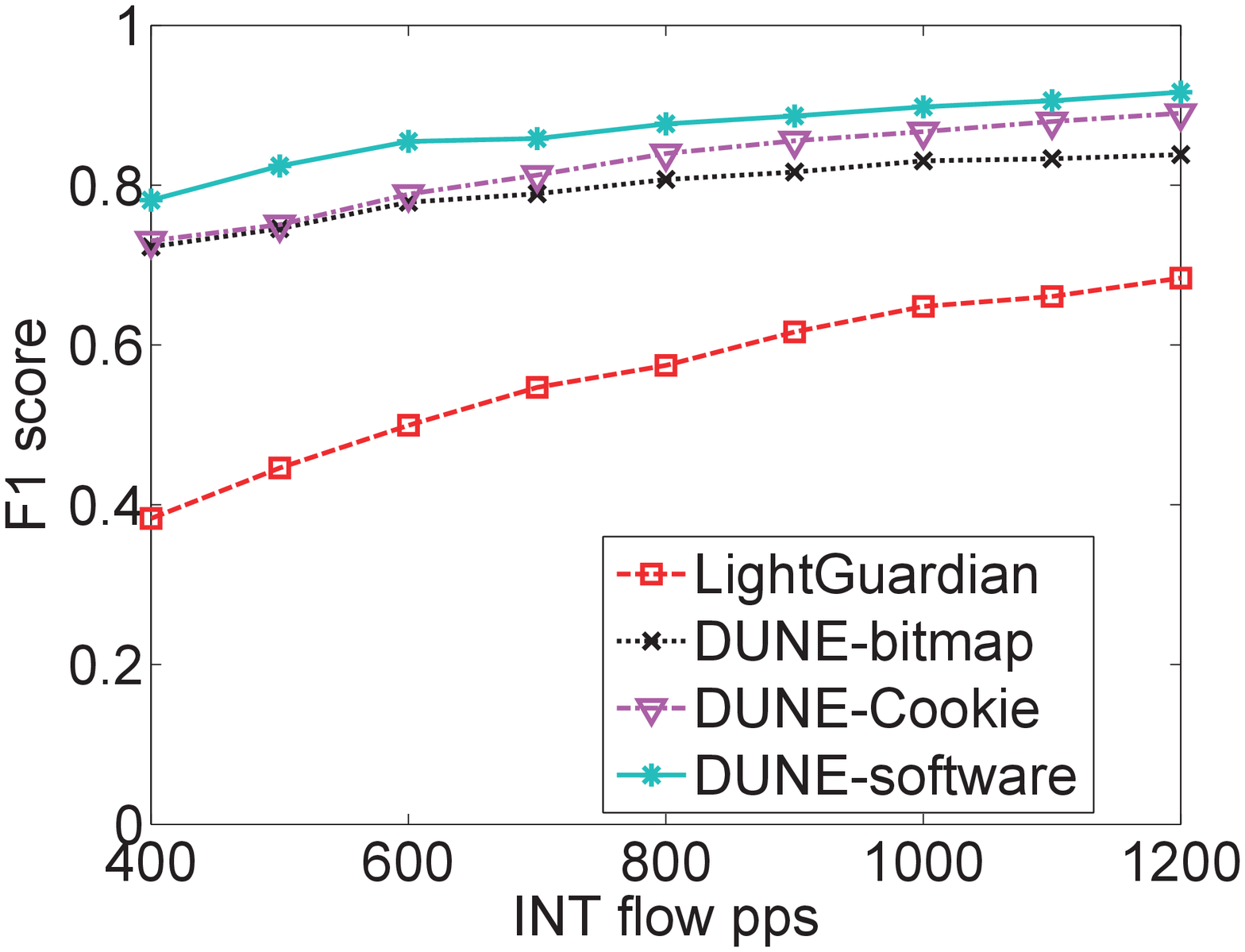}}
  \subfigure[Flow size distribution]{\includegraphics[height=1.33in]{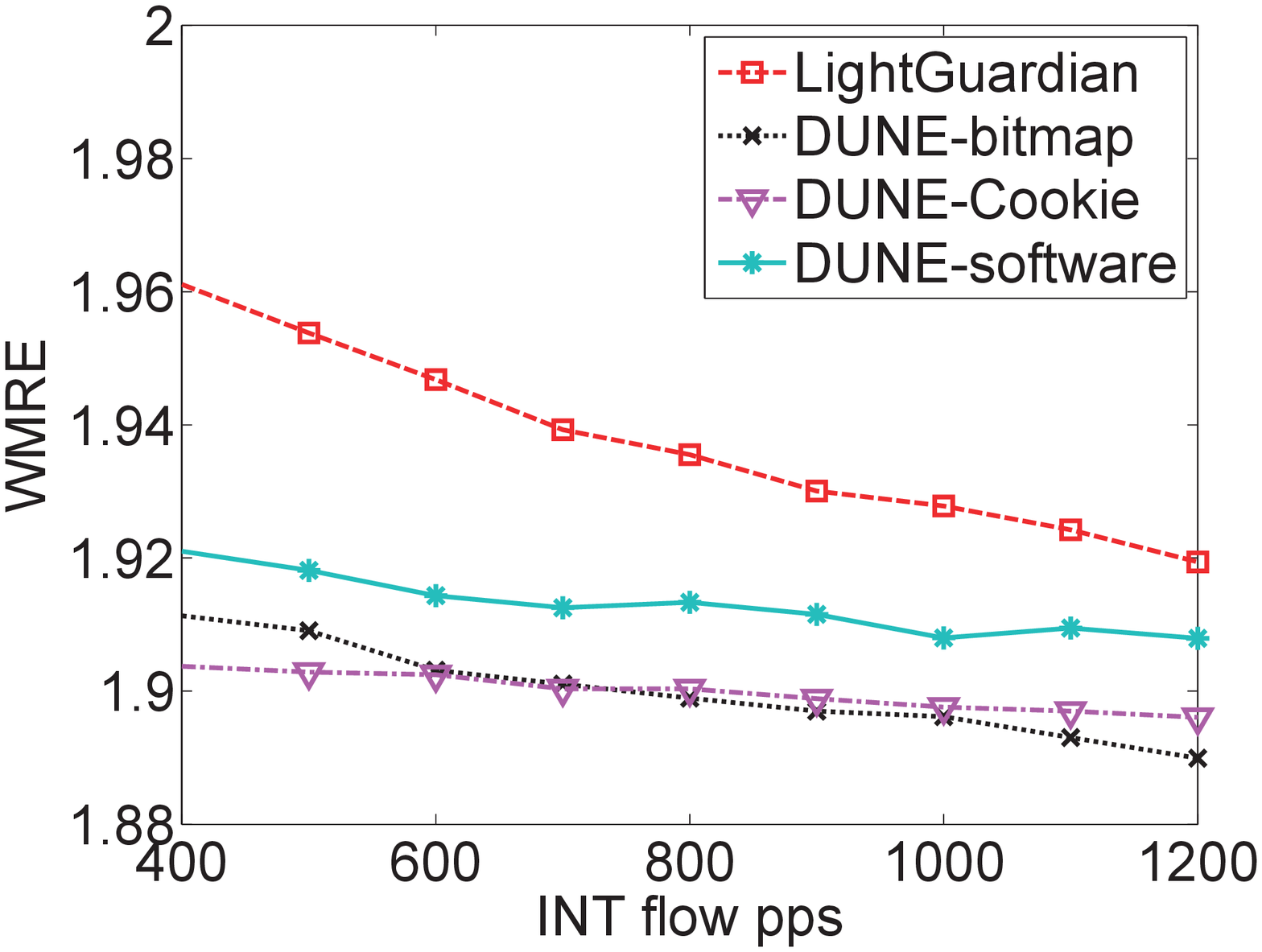}}
  \subfigure[Entropy]{\includegraphics[height=1.33in]{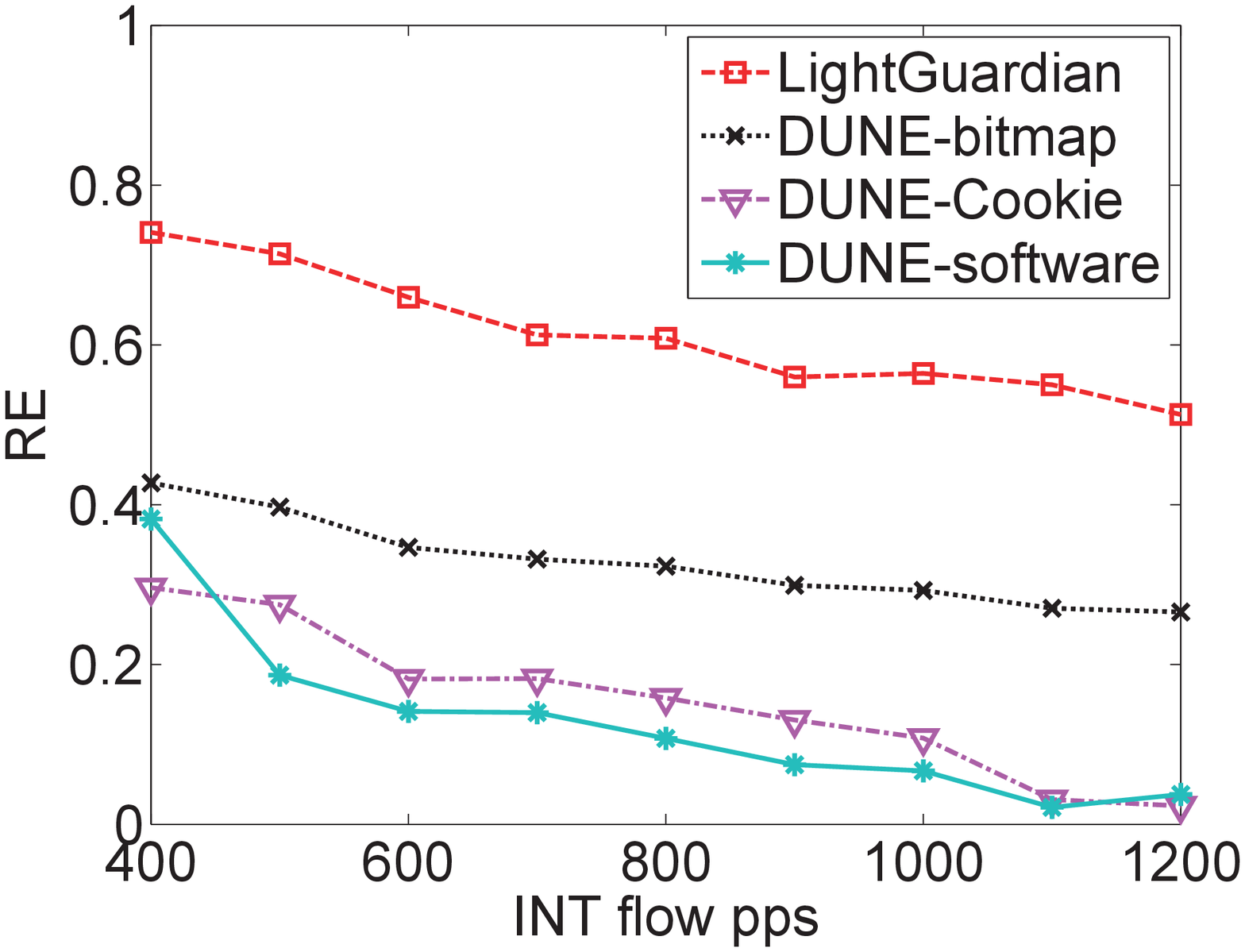}}
  \caption{Accuracies of (a) cardinality estimations in RE, (b) heavy hitter detections in F1-score, (c) flow size distribution estimations in WMRE, and (d) entropy estimations in RE with LightGuardian, DUNE-bitmap, DUNE-Cookie, and DUNE-software under various INT flow pps.}\label{Fig_TaskEval}
\end{figure*}

We conduct extensive experiments to evaluate DUNE, and in particular, we examine the four sketch-INT systems as the following.
\begin{itemize}
  \item \textbf{DUNE-bitmap}: In DUNE-bitmap, we employ the scatter sketchlet as described in Sec. \ref{Sec_Scatter}, and use the bitmap algorithm in Algorithm \ref{Alg_Bitmap} to select sketch buckets to send to end-host.
  \item \textbf{DUNE-Cookie}: The DUNE-Cookie system employs the scatter sketchlet and applies the Cookie algorithm in Algorithm \ref{Alg_Cookie} to select sketch buckets to sketchlets.
  \item \textbf{DUNE-software}: This solution differs from DUNE-Cookie in that it uses the software switch algorithm in Algorithm \ref{Alg_Software} to select sketch buckets.
  \item \textbf{LightGuardian}\cite{ZhaoLightGuardian21}: As a representative sketch-INT system, LightGuardian adopts the column sketchlet, and employs an algorithm named \emph{$k$+chance} to select the sketch columns. In the $k$+chance algorithm, a programmable sketch maintains $k$ bit arrays, each containing $w$ bits. When a column index $addr$ is randomly selected, the switch sequentially inspects the bits indexed at $addr$ in each array: If an 0-bit is encountered, the switch sets the bit as $1$, and adds the sketch column at $addr$ to the sketchlet; If all the bits have already been set as $1$, the switch randomly selects an other column. Ideally with the $k$+chance algorithm, all the columns will be sent to end-host with a fair chance.
\end{itemize}

%Besides the Tofino switch,
We have implemented DUNE-bitmap and DUNE-Cookie with Tofino switches. We also implement DUNE-bitmap, DUNE-Cookie, and LightGuardian on \texttt{bmv2} \cite{BMV2}, which is a P4-programmable software switch. We emulate DUNE-software with a standalone software switch, as its bucket selection algorithm is proactive, which can not be realized within an RMT pipeline.

Unless otherwise specified, in the following experiments, we set the sketch/bitmap/Cookie size as $d=2$ rows and $w=2^{15}$ columns. A sketch bucket contains $c=64$ bits, the size of a Cookie cell is $b=8$ bits, and the length of a scatter sketchlet's offset is $r=6$ bits. For the DUNE-Cookie system, we set the two threshold parameters as $\alpha=0.5$ and $\beta=1.0$, and for DUNE-software, we set $\phi=50$ and $\varphi=100$. In both DUNE-Cookie and DUNE-software, we set $s=1$, which means that a Cookie cell's value is halved each time the associated sketch bucket gets selected.
For evaluating LightGuardian, we employ $k=8$ bit arrays, which means that the bit arrays used by the $k$+chance algorithm consume four times switch memory comparing with the bitmap, or half of the memory comparing with the Cookie structure.

With the above parameter settings, the size of a column sketchlet should be at least $d\times c+\log_{2}w=143$ bits, and the minimum size of a scatter sketchlet should be $d\times c+\log_{2}w+d\times r=155$ bits. We can see that DUNE is lightweight by increasing LightGuardian's INT overhead no more than $8.4\%$ in our experiments.

We use the public available MAWI packet trace \cite{MAWI} captured from the WIDE backbone to drive the experiments. The trace contains $9.6$M flows, and we randomly select $6,000$ flows from the top-$50$K largest flows for each experiment.

\subsection{Measurement Accuracy}\label{Sec_TaskEval}

%We first evaluate measurement accuracies of various sketch INT systems.

\subsubsection{Measurement tasks and metrics}

We first conduct a number of network measurement tasks with the four sketch-INT systems. The tasks are:
\begin{itemize}
    \item \textbf{Cardinality estimation}. In this task, we count number of the distinct flows appear in the end-host reconstructed sketch to estimate the traffic cardinality.
    \item  \textbf{Heavy hitter detection}. This task aims to identify the top 10-\% largest flows with the reconstructed sketch at the end-host.
    \item \textbf{Flow size distribution estimation}. This task aims to estimate $m_i$, the number of the flows of size $i$ for all the possible sizes with the reconstructed sketch.
    \item \textbf{Entropy estimation}. This task estimates the entropy of the flows, which is defined as
   \[
   Entropy=\sum_{i}\left( i\times \frac{m_{i}}{M}\times \log \frac{m_{i}}{M}\right)
   \]
   where $m_i$ is the number of flows of size $i$ and $M=\sum_{i}m_i$, with the reconstructed sketch.
\end{itemize}

We use the following metrics to evaluate the measurement accuracies.
\begin{itemize}
\item \textbf{Relative error (RE)}: We use the relative error, which is defined as
\[
RE=\frac{\left\vert Estimated-Truth\right\vert }{Truth}
\]
to evaluate the cardinality and entropy estimations' accuracies.
\item  \textbf{F1-score}. For detecting heavy hitters, we use the F1-score to evaluate the accuracy.
\item \textbf{Weighted mean relative error (WMRE)}: For comparing the estimated flow size distribution with the ground truth, we compute WMRE as
    \[
    WMRE=\frac{\sum_{i=1}\left\vert m_{i}-\hat{m}_{i}\right\vert }{\sum_{i=1}\left( \frac{m_{i}+\hat{m}_{i}}{2}\right) }
    \]
    where $m_i$ and $\hat{m}_{i}$ are the estimated and ground-truth numbers of the flows of size $i$.
\end{itemize}

\subsubsection{Results}

We conduct the measurement tasks with the four sketch-INT systems, and present the results in Fig. \ref{Fig_TaskEval}.
%we presents the results of the measurement tasks conducted on the variants of the DUNE system as well as on LightGuardian.
In each experiment, we vary the INT flow's packet-per-second (pps), which determines the maximum number of the sketchlets that can be transferred from the switch sketch to the end-host, from $400$ to $1,200$ in the experiments.

We make several interesting observations for Fig. \ref{Fig_TaskEval}. The first observation is that all the Sketch-INT systems have better performances as the INT flow pps increases. This is easy to understand, as a higher pps indicates that the INT flow can bring more buckets to reconstruct the sketch at the end-host.

The second observation is that our proposed systems, i.e., DUNE-bitmap, DUNE-Cookie, and DUNE-software, are more accurate in all the tasks than LightGuardian. This is also easy to understand, as with the scatter sketchlet and bucket selection algorithms, our proposed systems actually deliver more measurement data of higher qualities to the end-host than LightGuardian.

The third observation is that among our proposed systems, there is no ``silver bullet'' for all the measurement tasks: a system may have good performance in one task, but may have poor accuracy in another.
For example in the cardinality estimation as in Fig. \ref{Fig_TaskEval}(a), DUNE-Cookie achieves the best accuracy for two reasons: 1) DUNE-Cookie halves a Cookie cell after selecting the corresponding sketch bucket, thus can avoid selecting the same bucket within a short time, as it takes time for the flow to grow,
but in DUNE-bitmap, a bucket can be selected again right after it is updated just once.
2) DUNE-Cookie selects sketch bucket within a limited range of $[addr, (addr+2^r-1)]$ starting from a random $addr$, thus can avoid repeatedly selecting a small number of  buckets that are updated very frequently, but with DUNE-software, which searches sketch buckets globally and exhaustively, the algorithm may repeatedly select a few very frequently updated buckets, while ignores the others.

For the other tasks, Fig. \ref{Fig_TaskEval}(b) shows that DUNE-software has the highest F1-scores in detecting the heavy hitters, as it employs a global and exhaustive search algorithm, and DUNE-Cookie, which also traces sketch buckets' update frequencies, better detects large flows than DUNE-bitmap.

For estimating the flow size distribution as in Fig. \ref{Fig_TaskEval}(c), DUNE-bitmap outperforms the other two systems for the reason that, it uniformly selects sketch buckets of all the flows; while with the global and exhaustive searching algorithm, DUNE-software is heavily biased towards the frequently updated sketch buckets, while ignores many small flows, thus derives a distorted flow size distribution .

Finally in the entropy estimation as in Fig. \ref{Fig_TaskEval}(d), since the entropy definition is biased towards large flows, whose sketch buckets have higher chances to be selected by DUNE-Cookie and DUNE-software, DUNE-software achieves the lowest RE, while the unbiased DUNE-bitmap system has the highest error rates among our proposed systems.

In summary, the experiment results in Fig. \ref{Fig_TaskEval} confirm that our proposed sketch-INT systems outperform the existing solution, thanks to the novel design of the scatter sketchlet and the smart bucket selection algorithms. Moreover, our proposed methods have differentiated performances in various measurement tasks, suggesting that it is important to choose the right method for each measurement task.
%because of the different behaviors of the bucket selection algorithms they employ.

%LightGuardian, and different systems are good at different measurement tasks. In particular, DUNE-software has the better performance in the tasks that are biased towards large flows, such as heavy hitter detection and entropy estimation, while DUNE-bitmap achieves the lowest error rates in the tasks of cardinality and flow size distribution estimations, in which flows are treated equally regardless of their sizes.

\subsection{Decomposing Measurement Errors}

\begin{figure}[tbp]
  \centering
  \subfigure[]{\includegraphics[height=1.30in]{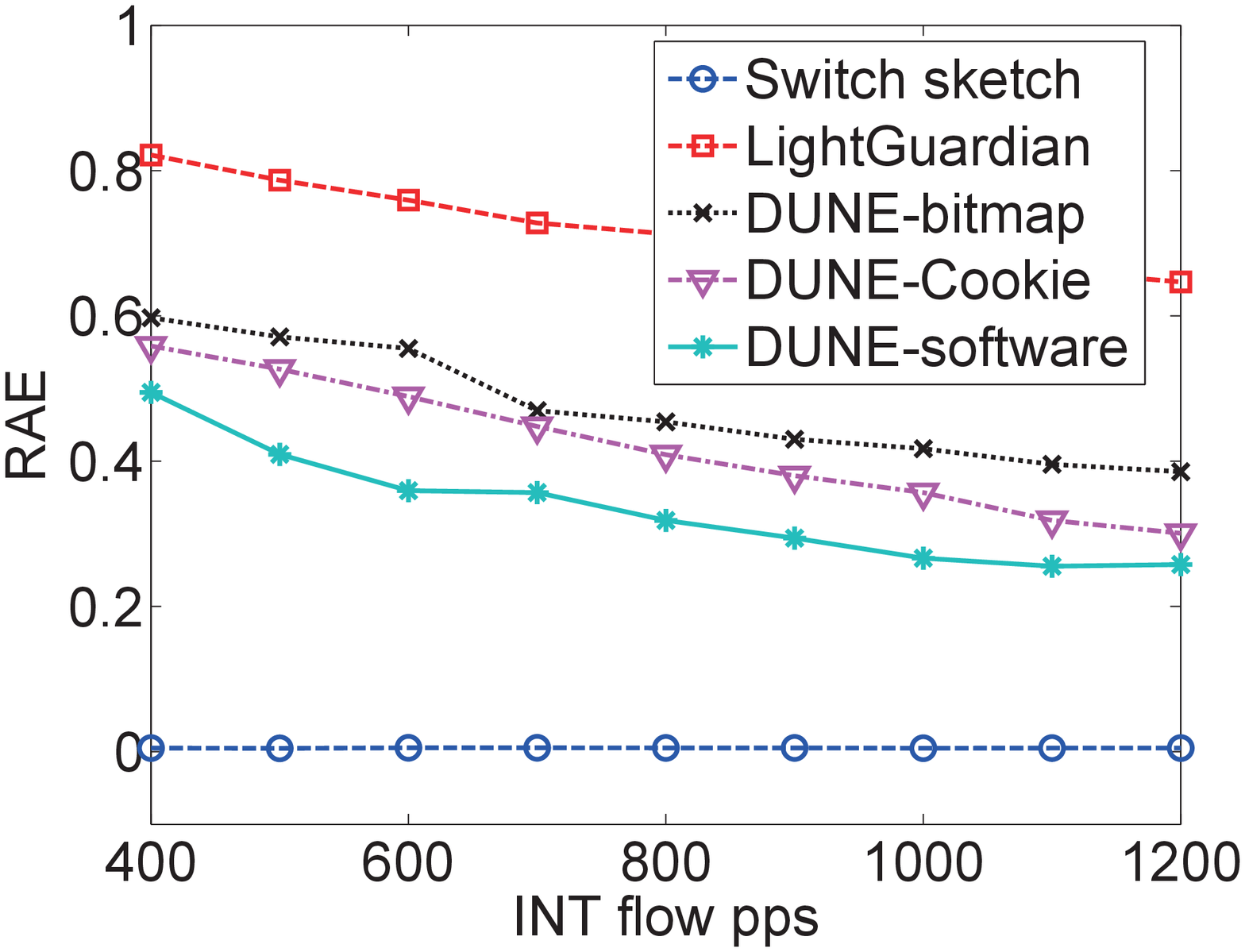}}
  \subfigure[]{\includegraphics[height=1.30in]{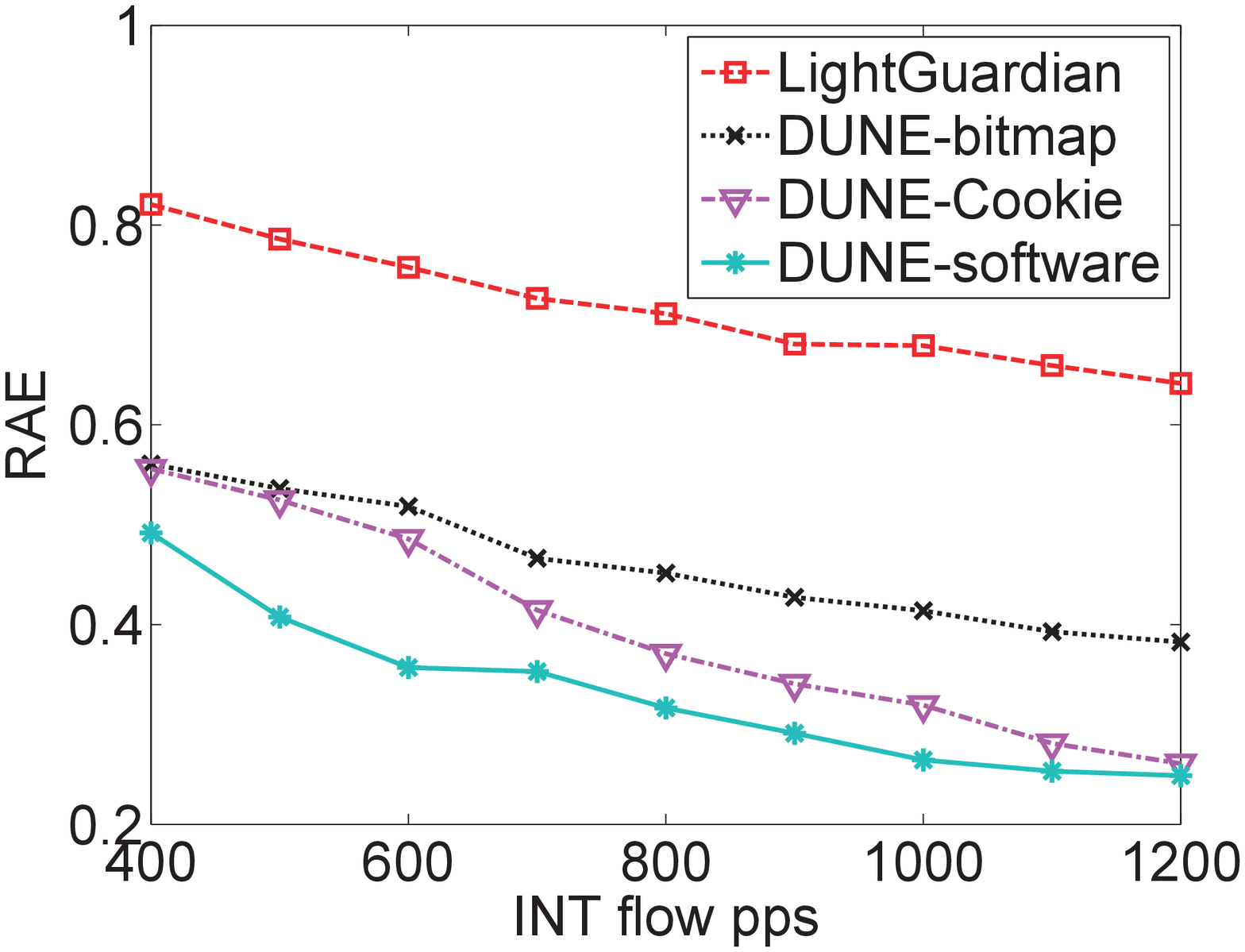}}
  \caption{(a) $RAE(\mathbf{n},\mathbf{n_{\mathbf{A}}|F})$ of switch sketch and $RAE(\mathbf{n},\mathbf{n_{\mathbf{A'}}|F})$ of reconstructed sketches by LightGuardian, DUNE-bitmap, DUNE-Cookie, and DUNE-software; (b) $RAE(\mathbf{n_{\mathbf{A}}},\mathbf{n_{\mathbf{A'}}|F})$ of reconstructed sketches by LightGuardian, DUNE-bitmap, DUNE-Cookie, and DUNE-software.}\label{Fig_RAE}
\end{figure}

The evaluation results in Fig. \ref{Fig_TaskEval} suggests that errors exist in the reconstructed sketch at end-host. An error may be caused by two different reasons: First, the error is caused by hash collision in the switch sketch, and the erroneous data is transferred to the end-host by sketchlets. Second, the measurement data in the switch sketch is error-free, but the reconstructed sketch at the end-host is not timely synchronized with the switch sketch, thus introduces inaccuracies because of the invalid or stale data in the sketch buckets, as we have seen in Sec. \ref{Sec_Motivation}. In the following, we seek to identify and quantify the two kinds of errors.

\subsubsection{Decomposing methods and metrics}
Before presenting our methods and metrics, we first introduce some notations. Let $\mathbf{x}=\{x_{f}|f\in \mathbf{F}\}$ be a set of measurement data on a network state (i.e., flow size) over a network flow set $\mathbf{F}$, and $\mathbf{y}=\{y_{f}|f\in \mathbf{F}\}$ be another set of measurement data on same state over a same flow set $\mathbf{F}$. We define the \emph{Relative Aggregated Error (RAE)} for comparing the measurement data $\mathbf{y}$ against $\mathbf{x}$ as
\begin{equation}
RAE(\mathbf{x},\mathbf{y|F})=\frac{\sum_{f\in \mathbf{F}}|x_{f}-y_{f}|}{\sum_{f\in \mathbf{F}}x_{f}}
\end{equation}

We use the following metrics to quantify the errors from different sources.
\begin{itemize}
  \item $RAE(\mathbf{n},\mathbf{n_{\mathbf{A}}|F})$: It is the RAE for comparing the measurement data in the switch sketch  against the ground truth, where $\mathbf{n}=\{n_{f}|f\in \mathbf{F}\}$ is the set of the ground truth flow size, and $\mathbf{n_{\mathbf{A}}}=\{n^{\mathbf{A}}_{f}|f\in \mathbf{F}\}$ is the flow sizes estimated by the switch sketch $\mathbf{A}$.
  \item $RAE(\mathbf{n},\mathbf{n_{\mathbf{A'}}|F})$: It is the RAE for comparing the measurement data in the reconstructed sketch at the end host against the ground truth, where $\mathbf{n_{\mathbf{A'}}}=\{n^{\mathbf{A'}}_{f}|f\in \mathbf{F}\}$ is the flow sizes estimated by the reconstructed sketch $\mathbf{A'}$.
  \item $RAE(\mathbf{n_{\mathbf{A}}},\mathbf{n_{\mathbf{A'}}|F})$: It is the RAE for comparing the flow sizes estimated by the reconstructed sketch against the ones estimated by the switch sketch.
\end{itemize}

From the above definition, we can see that $RAE(\mathbf{n},\mathbf{n_{\mathbf{A}}|F})$ quantifies the errors caused by hash collisions in the switch sketch, $RAE(\mathbf{n_{\mathbf{A}}},\mathbf{n_{\mathbf{A'}}|F})$ measures the errors caused by the invalid and stale bucket data in the end-host reconstructed sketch, and $RAE(\mathbf{n},\mathbf{n_{\mathbf{A'}}|F})$ captures the overall errors.

\subsubsection{Results}

We run the four sketch-INT systems to estimate the sizes of $6,000$ flows from the MAWI trace, and compare RAEs of the different systems in Fig. \ref{Fig_RAE}. In particular, we compare the flow sizes estimated by the switch sketch with the ground truth, and present $RAE(\mathbf{n},\mathbf{n_{\mathbf{A}}|F})$ denoted as ``switch sketch'' in Fig. \ref{Fig_RAE}(a); we also compare the end-host sketches reconstructed by different systems against the ground truth in $RAE(\mathbf{n},\mathbf{n_{\mathbf{A'}}|F})$ in the figure; In Fig. \ref{Fig_RAE}(b) we compare the end-host reconstructed sketches against the switch sketch, and present $RAE(\mathbf{n_{\mathbf{A}}},\mathbf{n_{\mathbf{A'}}|F})$ of the four systems.

\begin{figure}
  \centering
  \includegraphics[width=2.0in]{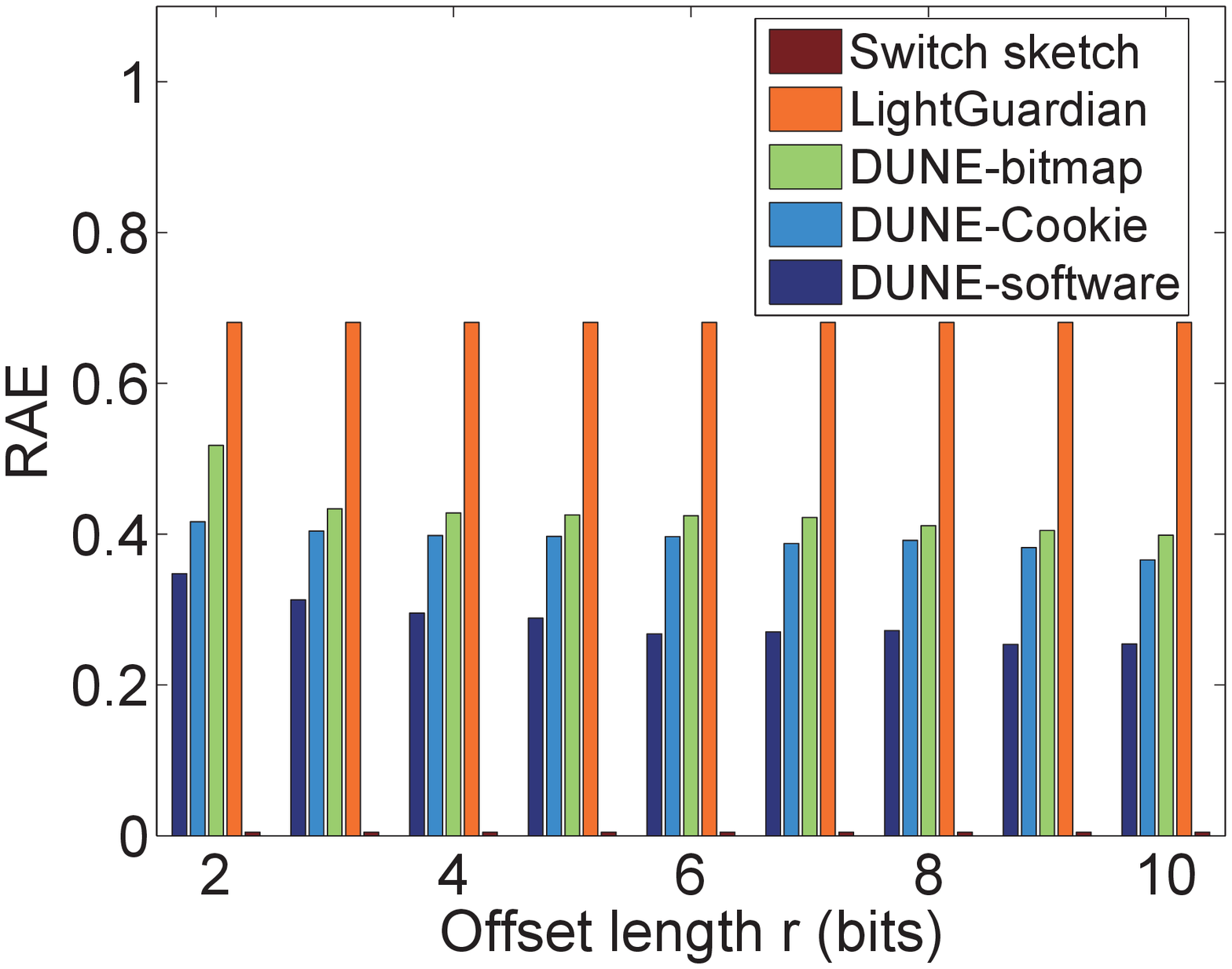}\\
  \caption{Impact of scatter sketchlet offset length $r$.}
  \label{Fig_OffsetLen}
\end{figure}

From Fig. \ref{Fig_RAE} we can make several observations. First, the reconstructed sketches at end-hosts contain much more errors comparing with the switch sketch, suggesting that most of the inaccuracies arise from the invalid and stale bucket data in the reconstructed sketch.
Second, a larger INT flow pps can considerably reduce the errors in the reconstructed sketches. For example, the values of $RAE(\mathbf{n},\mathbf{n_{\mathbf{A'}}|F})$ are reduced $21.8\%$, $31.8\%$, $53.1\%$, and $49.5\%$ in LightGuardian, DUNE-bitmap, DUNE-Cookie, and DUNE-software respectively, when pps is increased from $400$ to $1,200$.
Third, DUNE-software has the lowest RAE because of its global and exhaustive searching algorithm; while DUNE-bitmap has the highest RAE as it selects sketch buckets in an unbiased way.
%Note that a lower RAE does not necessarily mean that the system is superior, as we have seen in Sec. \ref{Sec_TaskEval}.
The last observation is that thanks to the novel design of the scatter sketch and the smart bucket selection algorithms, our proposed systems achieve much lower error rates than LightGuardian. For example, under the $1,200$ INT flow pps, DUNE-bitmap, DUNE-Cookie, and DUNE-software reduce LightGuardian's $RAE(\mathbf{n},\mathbf{n_{\mathbf{A'}}|F})$ $40.4\%$, $59.4\%$, and $61.2\%$ respectively.

\subsection{Impact of Offset Length $r$}

The design of the scatter sketchlet enables a switch to select buckets that contain ``fresh'' measurement data in a range of $[addr, (addr+2^r-1)]$.
%allows a sketch bucket in the range of $[addr, addr+(2^r-1)]$ to be selected, therefore
Intuitively, the larger the offset length $r$ is, the higher chance that a sketch bucket with ``fresh'' measurement data can be selected, and the higher estimation accuracy the reconstructed sketch can achieve.

In this experiment, we run different sketch-INT systems under various offset length $r$ ranging from $2$ to $10$ bits, and present $RAE(\mathbf{n_{\mathbf{A}}},\mathbf{n_{\mathbf{A'}}|F})$ of LightGuardian, DUNE-bitmap, DUNE-Cookie, and DUNE-software in Fig. \ref{Fig_OffsetLen}. We also plot $RAE(\mathbf{n},\mathbf{n_{\mathbf{A}}|F})$ of the switch sketch for comparison.

From the figure we can see that increasing the offset length do reduce the errors, but the reduction is not very significant. For example, for DUNE-Cookie, by increasing $r$ from $2$ to $10$ bits, the error reduction ratio is $5.6\%$. Recall that in the Tofino implementation as described in Sec. \ref{Sec_Implement}, a packet accesses $2^r$ registers for selecting a sketch bucket. The result in Fig. \ref{Fig_OffsetLen} suggests that even with a smaller offset length (and consequently, fewer registers), the systems of DUNE-bitmap and DUNE-Cookie can still have decent accuracies comparing with LightGuardian. On the other hand, with the DUNE-software system that runs on software switches, a higher accuracy can be expected by exploiting a larger offset length $r$.

\subsection{Impact of Cookie Cell Size $b$}

\begin{figure}
  \centering
  \includegraphics[width=2.0in]{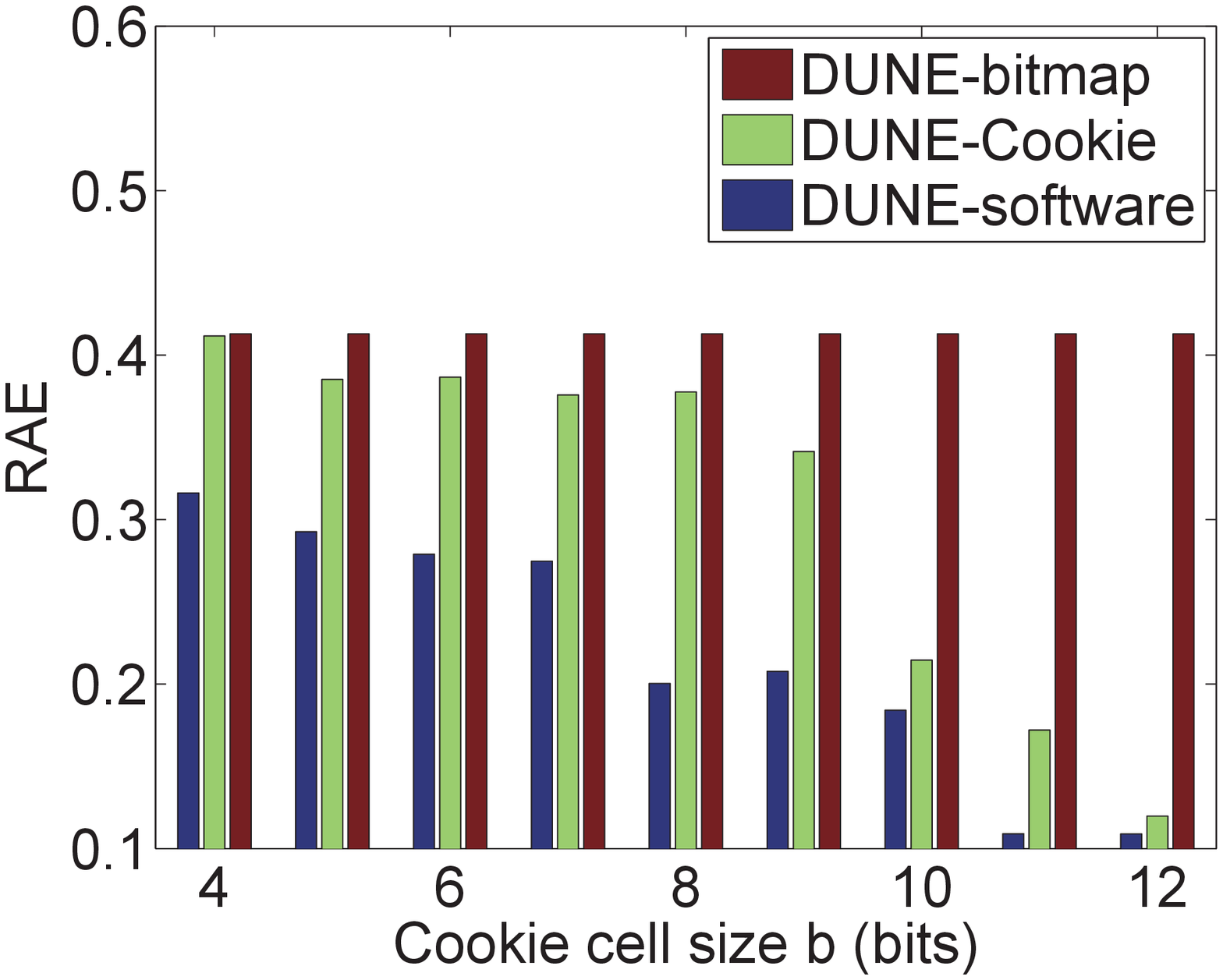}\\
  \caption{Impact of Cookie cell size $b$.}
  \label{Fig_OffsetLen}
\end{figure}

In the DUNE-Cookie or DUNE-software system, we place a Cookie data structure, which has an identical logical structure as the sketch, to trace the ``freshness'' of the measurement data in sketch buckets. In this experiment, we vary the size of a Cookie cell from $4$ to $12$ bits, and present $RAE(\mathbf{n_{\mathbf{A}}},\mathbf{n_{\mathbf{A'}}|F})$ of the DUNE-Cookie and DUNE-software systems in Fig. \ref{Fig_OffsetLen}. We also plot DUNE-bitmap's $RAE(\mathbf{n_{\mathbf{A}}},\mathbf{n_{\mathbf{A'}}|F})$ for comparison.

From Fig. \ref{Fig_OffsetLen}, we can see that by increasing the Cookie cell size, better accuracies can be achieved by the reconstructed sketches in DUNE-Cookie and DUNE-software, in particular, when $b$ exceeds $9$ bits, the errors are considerably reduced. We believe this is because with the MAWI traffic trace, the network flows that grow very fast can be accurately identified by the Cookie algorithm when the Cookie cells are capable to trace up to $512$ updates.
%, the algorithms can better differentiate the fast growing flows from the ordinary ones, and select their corresponding sketch buckets accurately.
Our observation suggests that a tradeoff is allowed between the memory usage and the measurement accuracy: for hardware switches such as the Tofino switch that lacks memory resource, a small Cookie size can provide reasonable accuracy, while on software switches, we can pursue a higher accuracy at a cost of a larger memory usage.

\subsection{Forwarding Performance}

\begin{figure}[tbp]
  \centering
  \subfigure[]{\includegraphics[height=1.30in]{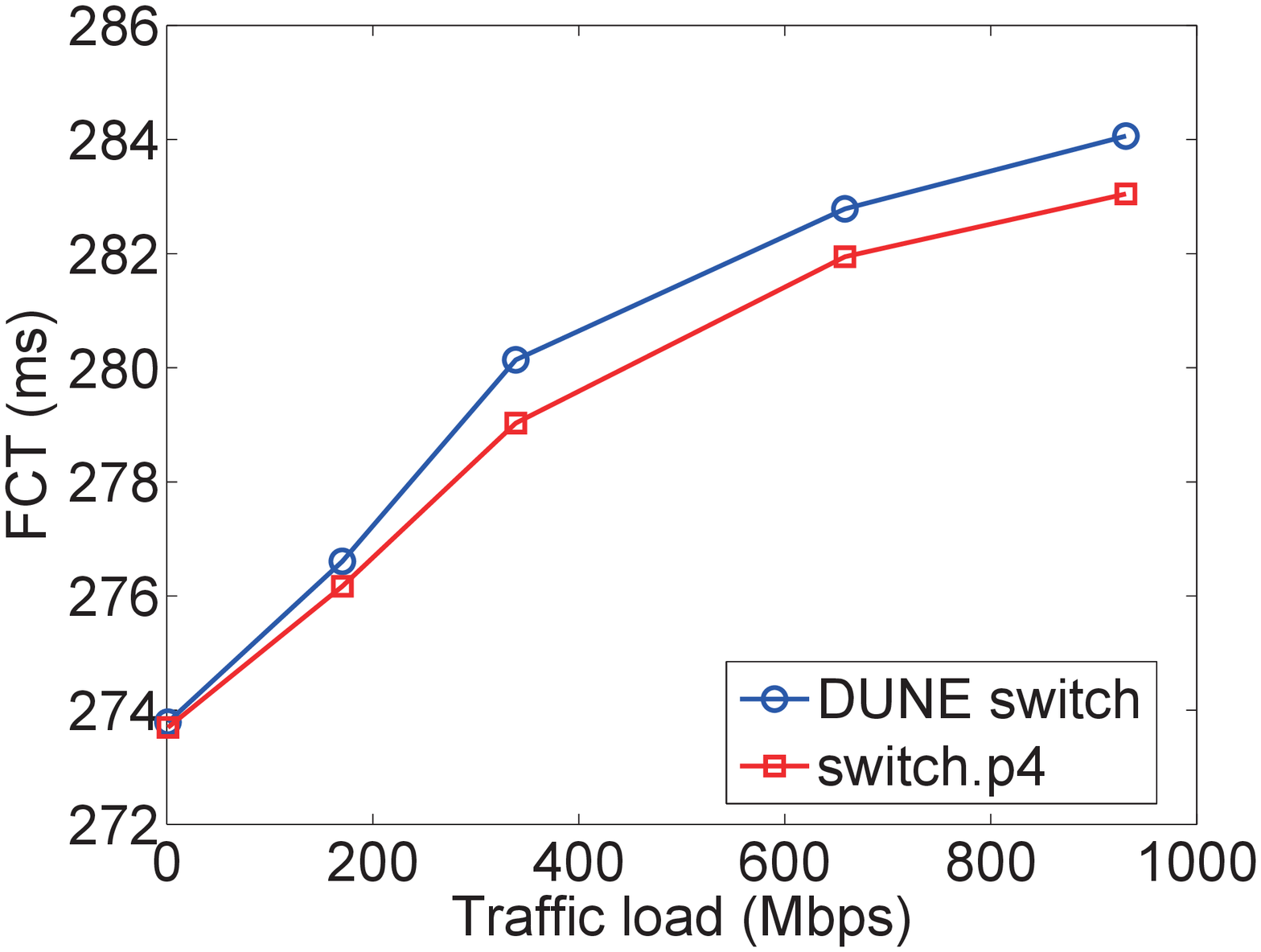}}
  \subfigure[]{\includegraphics[height=1.30in]{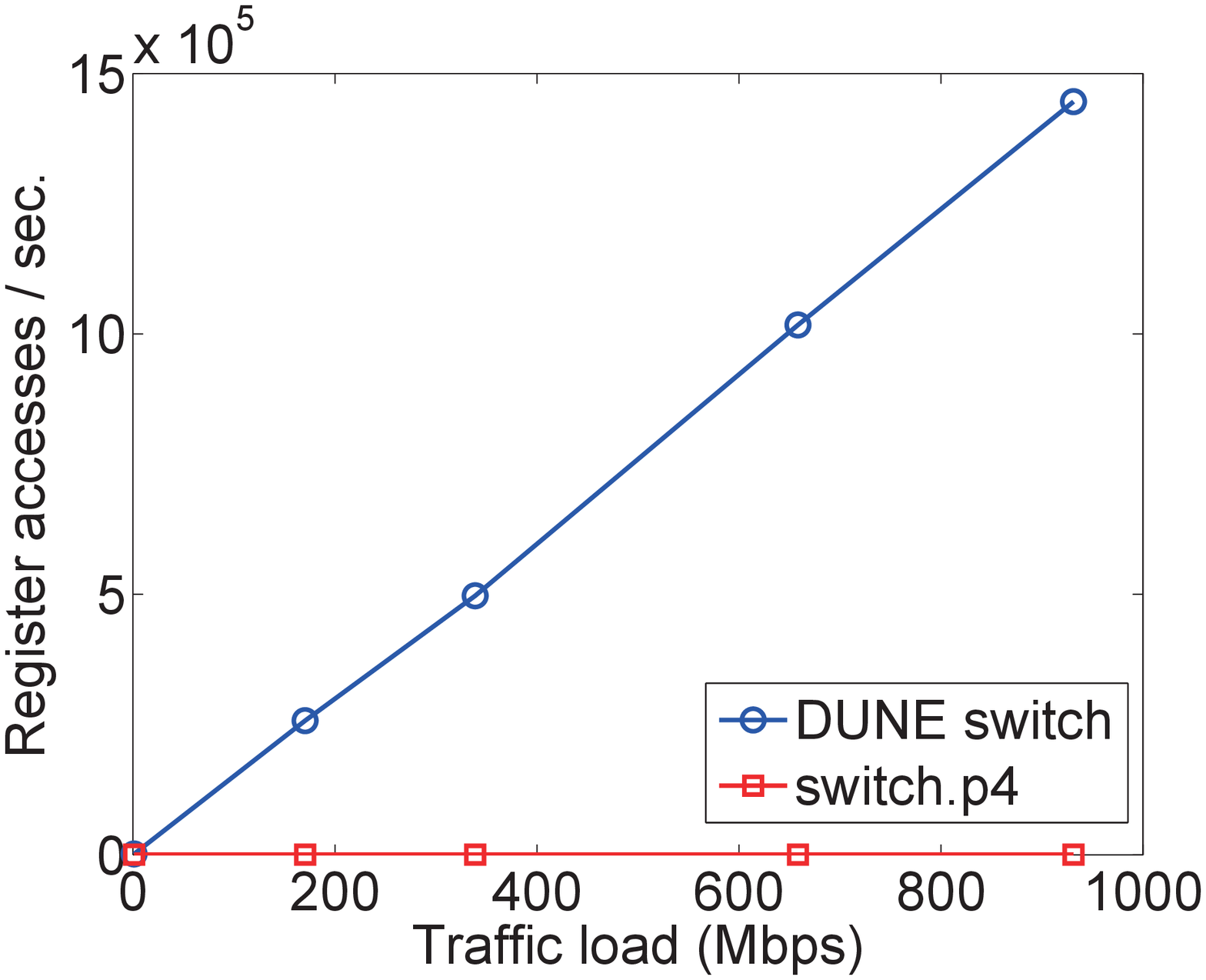}}
  \caption{(a) FCTs of 1000-packet INT flow under various traffic load forwarded by DUNE switch and baseline switch. (b) Register accesses per second in DUNE switch and baseline switch under various traffic load.}\label{Fig_Performance}
\end{figure}

We evaluate the performance of the DUNE prototype implemented on the Barefoot Tofino switch. As described in Sec. \ref{Sec_Implement}, the realized sketch/bitmap/Cookie has $d=2$ rows and $w=2^{15}$ columns. We set the offset length as $r=3$, and each bitmap/Cookie row is realized with $2^r=8$ registers.

The major difference between a DUNE switch and a conventional L2/L3 switch is that in a DUNE switch, a packet of an ordinary network flow is required to access $2\times d=4$ registers to update the sketch buckets as well as the bit/Cookie cells, and an INT flow packet needs to access $d+d\times 2^r=18$ registers to select a sketch bucket by inspecting bits/Cookie cells in a range of $[addr,(addr+2^r-1)]$. Our concern is, \emph{will the register accesses significantly impact the switch's forwarding performance?}

In our experiment, we send a traffic workload varying from \SI{1}{\mega\bit\per\second} to \SI{950}{\mega\bit\per\second} from the MAWI trace to the DUNE switch, and use an INT flow containing $1,000$ packets to carry the sketchlets to the end-host. We measure the INT flow's flow completion time (FCT) to evaluate the switch's forwarding performance, and plot the results in Fig. \ref{Fig_Performance}(a). For comparison, we also run \texttt{switch.p4} \cite{SwitchP422}, a baseline L3 switch implementation on the Tofino switch, under same traffic loads and plot the FCTs of a same $1,000$-packet flow in the figure.
Fig. \ref{Fig_Performance}(b) presents the averaged register accesses in the two switches. Note that when forwarding a packet, the baseline switch does not access any register.

From the figure we can see that the flows traversing the DUNE switch have FCTs only slightly longer comparing with the ones of the baseline switch, despite that the DUNE switch make a large number of register accesses. For example, even under the highest  traffic load in our experiment, register accesses only prolong the FCT less than $0.36\%$, and the performance can be further improved \cite{NamkungSketchLib22}.
%Note that we implement DUNE naively. If advanced optimizations techniques \cite{NamkungSketchLib22} are applied in our implementation, the loss of the switch's packet forwarding rate can be further reduced.
The experiment result suggests that our proposed sketch-INT system is practical to be deployed in production networks to handle real-world network traffics.

\section{Conclusion}\label{Sec_Conclude}

In this paper, we presented \emph{DUNE}, a lightweight and accurate sketch-INT network measurement system. DUNE follows the ``reconstructing sketch at end-host'' approach in combining sketch and INT, thus is lightweight regarding network overhead. To combat the inaccuracies caused by the invalid and stale data in the buckets of the end-host reconstructed sketch, we made two innovations: First, we designed a novel \emph{scatter sketchlet} that allows a switch to select individual buckets to add to sketchlet; Second, we developed data structures for tracing ``freshness'' of the data in sketch buckets, and proposed algorithms for smartly selecting buckets to send to end-host. We theoretically proved that our proposed methods have higher efficiency in transferring measurement data, and better adapt to skewed network flows. We implemented a prototype system on P4-programmable Tofino switches under the switch's register access constraints. We extensively evaluated our proposed system with experiments driven by real-world backbone traffic, and showed with the experiment results that DUNE can significantly improve the measurement accuracies by avoiding up to $60\%$ errors, while only slightly reduce the switch's forwarding rate less than $0.36\%$.

\bibliographystyle{IEEEtran}
\bibliography{IEEEabrv, TechRep}

\end{document}